\documentclass{article}
\usepackage{graphicx}
\usepackage{amsfonts}
\usepackage{amsmath}
\usepackage{amssymb}
\usepackage{url}
\usepackage{fancyhdr}
\usepackage{indentfirst}
\usepackage[misc]{ifsym}
\usepackage{enumerate}
\usepackage{tikz}
\usetikzlibrary {patterns,patterns.meta}
\usepackage{slashed}
\usepackage{dsfont}
\usepackage[colorlinks=true,citecolor=blue]{hyperref}
\usepackage{amsthm}
\usepackage{color}
\definecolor{r}{rgb}{1,0,0}
\usepackage{natbib}
\usepackage{comment}
\usepackage{capt-of}
\def\d{\mathrm{d}}
\def\laweq{\buildrel \mathrm{d} \over =}

\DeclareMathOperator*{\argmin}{arg\,min}

\newcommand{\ES}{\mathrm{ES}}

\newcommand{\VaR}{\mathrm{VaR}}

\newcommand{\E}{\mathbb{E}}

\newcommand{\R}{\mathbb{R}}
\newcommand{\N}{\mathbb{N}}
\newcommand{\p}{\mathbb{P}}

\newcommand{\id}{\mathds{1}}

\renewcommand{\(}{\left(}
\renewcommand{\)}{\right)}
\renewcommand{\ge}{\geqslant}
\renewcommand{\le}{\leqslant}
\renewcommand{\geq}{\geqslant}

\renewcommand{\epsilon}{\varepsilon}
\newcommand{\esssup}{\mathrm{ess\mbox{-}sup}}
\newcommand{\essinf}{\mathrm{ess\mbox{-}inf}}

\theoremstyle{plain}
\newtheorem{theorem}{Theorem}
\newtheorem{corollary}{Corollary}
\newtheorem{lemma}{Lemma}
\newtheorem{proposition}{Proposition}
\theoremstyle{definition}
\newtheorem{definition}{Definition}
\newtheorem{example}{Example}

\theoremstyle{remark}
\newtheorem{remark}{Remark}

\newcommand{\cet}{\begin{center}}
\newcommand{\ecet}{\end{center}}

\usepackage{setspace}

\begin{document}

\title{Universal Value‐at‐Risk superadditivity}
%\date{}
\author{
Yuyu Chen\thanks{Department of Economics, University of Melbourne,  Australia. \Letter~{\scriptsize\url{yuyu.chen@unimelb.edu.au}}}
\and Liyuan Lin\thanks{Department of Econometrics and Business Statistics,
 Monash University, Australia.
\Letter~{\scriptsize\url{Liyuan.Lin@monash.edu}}}
  \and Ruodu Wang\thanks{Department of Statistics and Actuarial Science, University of Waterloo, Canada.
  \Letter~{\scriptsize\url{wang@uwaterloo.ca}}}}
\maketitle
 
\begin{abstract}
Value-at-Risk (VaR) is a standard regulatory risk measure, and its failure of subadditivity is well known. Much less appreciated is that for sufficiently heavy-tailed losses, VaR can be superadditive uniformly across all probability levels, a phenomenon strictly stronger than the asymptotic superadditivity studied in extreme value theory. We call this property universal VaR superadditivity (UVS). We study UVS and its stronger weighted version (WUVS)  as properties of random vectors rather than of marginal distributions. This perspective unifies and extends a recent line of work on iid infinite-mean models. UVS, except for trivial cases, imposes an infinite-mean structure. We establish preservation properties of UVS and WUVS under increasing and convex transformations, weak convergence, and certain distributional mixtures, and use these tools to prove UVS and WUVS for non-identically distributed risks in several large families including completely subscalable, super-Cauchy, and inverted subadditive risks, extending results previously available only in the iid case.  In many results, we also establish strict versions of UVS and WUVS, which lead to stronger decision-theoretic implications. As a consequence, for any portfolio satisfying WUVS, every distortion risk measure is superadditive, so an optimal allocation concentrates on a single asset, and diversification is never beneficial.

% Value-at-Risk (VaR) is a regulatory risk measure in finance and insurance, yet it is well known to have theoretical drawbacks, most notably the lack of subadditivity.
% For risk models without finite mean, VaR may even exhibit superadditivity across all probability levels. We refer to this phenomenon as universal VaR superadditivity (UVS). The paper offers a comprehensive study  on UVS, with several general properties. UVS has strong implications in decision making, reflecting the possible negative effects of diversification  for heavy-tailed risks. While several technical results in the literature yield UVS for identically distributed risks, this paper studies UVS for general portfolio vectors. We establish UVS for three large classes of risks: completely subscalable risks, super-Cauchy risks, and a subclass of inverted subadditive risks, extending existing results. In addition, we study the relationships among several classes of risks that satisfy UVS.

\end{abstract}

\section{Introduction}

Value-at-Risk (VaR) is one of the most widely used risk measures in finance and insurance, partially due to its intuitive interpretation. For a probability level $p\in(0,1)$, often chosen as  $0.99$ or $0.995$, VaR is defined as the (left) quantile of the loss distribution, representing the smallest threshold such that losses exceed this value with probability at most $1-p$. 
Over the past three decades, it has been a standard tool in several regulatory frameworks, such as the Basel Accords for banking regulation and Solvency II for insurance regulation, alongside its sister risk measure, the Expected Shortfall (ES). 
Although  VaR has several limitations compared to ES, including the lack of coherence \citep{ADEH99} and insensitivity to tail risk \citep{ELW18}, it has many intrinsic advantages, and remains widely used in financial practice. We refer to \cite{embrechts2014academic} for a discussion on the regulatory use of VaR and ES.

The key theoretical property that underpins the discussions of VaR is its possible superadditivity, showing a lack of consistency with respect to diversification. 
The VaR subadditivity, or its counterpart, superadditivity, has been extensively studied in Extreme Value Theory \citep{EKM97} in the asymptotic sense (i.e., the asymptotic behavior of VaR as its probability level goes to $1$). For iid regularly varying risks, VaR is asymptotically subadditive (resp.~superadditive) if the tail parameter of the risks is larger (resp.~less) than 1; the results for dependent regularly varying risks can be found in \cite{ELW09} and \cite{mainik2010optimal}. Note that a regularly varying risk is integrable if its tail parameter is strictly larger than 1. Thus, the integrability of risks is crucial to the asymptotic subadditivity/superadditivity of VaR.
The asymptotic behaviors of VaR have also been extended to a range of probability levels for a class of regularly varying risks. For iid risks following the symmetric stable distribution and for probability levels larger than 0.5, \cite{I09} showed that VaR is subadditive/superadditive if the tail parameter of the stable risks is larger/smaller than 1.

Our paper offers 
a comprehensive study of \emph{universal VaR superadditivity} (UVS),
which is the property of a random vector such that superadditivity of VaR holds for all probability thresholds in $(0,1)$. 
This property was recently studied by \cite{M25} for nonnegative random variables.
We also introduce the related stronger notion of weighted universal VaR superadditivity (WUVS), as well as strict UVS and strict WUVS. These notions are formally defined in Section \ref{sec:setup}.
The concepts of UVS and WUVS involve both the marginal distributions and the dependence structure. They are intimately connected to 
an important phenomenon actively studied in the recent literature. That is, 
% when the marginal distributions are identical, 
UVS can be reformulated via a first-order stochastic dominance relation (see Lemma \ref{lem:st}). With this reformulation, 
  UVS of iid Pareto risks with infinite mean was established by \cite{CEW25}, and the result has been extended to different distributional assumptions by \cite{ALO25}, \cite{CS25}, \cite{Muller25}, \cite{M25}, and \cite{V25}. 
The studies above, with the exception of \cite{M25}, focus on identical marginal distributions or a location-scale family, and therefore, the results are often obtained for a single distribution instead of a random vector. Except for \cite{CEW25}, the above studies do not address strict UVS and WUVS, which are also investigated in this paper. 

General properties of UVS and WUVS are studied in Section \ref{sec:VaR}. 
A necessary condition for UVS, outside of trivial cases, is the existence of infinite mean. Applications of infinite-mean models in insurance claims and catastrophic losses can be found in, e.g., \cite{CW25}. Indeed, as shown by \cite{IK25}, UVS holds for integrable risks only if the risks are perfectly positively dependent, where VaR is additive. We generalize this observation and a related result of \cite{M25} in Proposition \ref{prop:como} to one-sided integrable random variables. 
 As our first main result, Theorem \ref{th:convex} shows that (strict) UVS and WUVS are preserved under  convex transformation in different senses.  Theorem \ref{pro:SAw} further shows that strict UVS and WUVS can be obtained by applying strictly convex transformations to nonstrict UVS and WUVS random vectors with positive density.
Closure properties of UVS and WUVS are obtained in Theorem \ref{th:closure},  showing that UVS and WUVS are closed under convergence in distribution and a particular form of distributional mixture.  Unlike the superadditivity of VaR over a specific range of $(0,1)$,  UVS provides very strong implications in decision making.
In particular,  UVS for a loss portfolio implies that any monotone and comonotonic-additive risk measure is also superadditive for the portfolio (Proposition \ref{prop:sup}). WUVS further implies that diversification does not have benefits for portfolio selection (Proposition \ref{prop:rm}). 

%Among the aforementioned studies on UVS, \cite{M25} is the only one that considered non-identically distributed risks; risks in the other papers are identically distributed and vary up to location-scale transforms. 
We proceed to establish UVS and WUVS for three large classes of independent and non-identically distributed risks in Section \ref{sec:ind}, generalizing existing results for iid risks. The first is the completely subscalable risk introduced by \cite{V25}, who focused on the comparison of  convex combinations of iid risks and their distribution mixtures. Theorem \ref{thm:cs} shows
that independent completely subscalable risks satisfy WUVS even when their marginal distributions
are not identical. Note that in the case of iid risks, the two maximal classes studied for WUVS to date are the inverted subadditive class of \cite{ALO25} and the super-Cauchy class of \cite{Muller25}; the  inverted subadditive class is larger than the class of completely subscalable risks. However, WUVS cannot be extended to arbitrary independent and non-identically distributed inverted subadditive risks, as shown by Example \ref{ex:counter-D}. Nevertheless, WUVS holds for independent super-Cauchy risks (Proposition \ref{cor:SC}) and for some independent inverted subadditive risks (Proposition \ref{cor:IS}).
UVS and WUVS are also shown for a class of negatively dependent risks introduced by \cite{CS25} (Proposition \ref{prop:D}). Section \ref{sec:relation} presents a comparison among the classes of risks that satisfy UVS and WUVS. Section \ref{sec:con} concludes the paper.

\subsection*{Notation} 
Let $[n]=\{1,\dots,n\}$ for $n\in\N$. For $x\in\R$,  define $(x)_+=\max(x,0)$ and $(x)_-=\max(-x,0)$. A function $f$ on $(0,\infty)$ is \emph{subadditive} (resp.~\emph{superadditive}) if $f(x+y)\le f(x)+f(y)$ (resp.~$f(x+y)\ge f(x)+f(y)$) for any $x,y>0$.
Let $f^{-1}$ and $f^{-1+}$   be the left and right inverse functions of a function $f: \R\to \R$, defined as 
\begin{align*}
    f^{-1}(y)&=\inf\{x \in \R: f(x)\ge y\} 
    \mbox{~~~and~~~} f^{-1+}(y)=\sup\{x\in \R: f(x)\le y\},
\end{align*} 
for $y \in \R$ with $\inf \varnothing=\infty$   and  $\sup \varnothing=-\infty.$

All random variables are defined on an atomless probability space $(\Omega,\mathcal F,\p)$, and $\mathcal L_n^0$ is the set of all $n$-dimensional random vectors on this space.  
If a random variable $X$ has a cumulative distribution function $F$, we write $X\sim F$ and use $\overline F=1-F$ to denote the survival function of  $X$.  The essential infimum and essential supremum of $X$ are defined as
$$\essinf X=\inf\{x\in\R: \p(X\le x)>0\}\mbox{~~~and~~~}\esssup X=\sup\{x\in\R: \p(X\le x)<1\}.$$
If two random variables $X$ and $Y$ have the same cumulative distribution function, we write $X\laweq Y$. We say a random vector $(X_1,\dots,X_n)$ is independent and identically distributed (iid) if $X_1,\dots,X_n$ are iid.

\section{Universal VaR superadditivity} \label{sec:setup}
 
For a random variable $ X\sim F$, the Value-at-Risk (VaR) at   level $p\in (0,1)$ is defined as the left  inverse function of $F$, that is,
  \begin{equation*}%\label{VaR}
  \VaR_{p}(X)=F^{-1}(p)=\inf\{x\in\R:F(x)\geq p\}.
  \end{equation*}
  Fix an integer $n\ge 2$. 
  We say a random vector $(X_1, \dots, X_n)$ satisfies  \emph{universal VaR superadditivity (UVS)}  if 
\begin{equation}\label{eq:main}   \VaR_{p}\left( \sum_{i=1}^n X_{i}\right) \ge  \sum_{i=1}^n\VaR_{p}(X_i) \mbox{~~for all~~} p\in(0,1);\end{equation}
  \emph{universal VaR subadditivity (resp.~additivity)} is defined by replacing $\ge$ in  \eqref{eq:main} with $\le$ (resp.~$=$). %holds for $X_1,\dots,X_n$. 
Throughout this paper, we will mainly focus on superadditivity; the case of subadditivity is symmetric. 
Clearly, any constant vector satisfies UVS (indeed, additivity), but that case is not interesting. For some applications, the strict version of UVS is needed.
We say that $(X_1,\dots,X_n)$
satisfies \emph{strict UVS} if \eqref{eq:main} holds   with strict inequality  for all $p\in(0,1)$.

  Strict UVS  is a very strong property; it rules out all finite-mean distributions.
Let $\mathcal{L}^{\rm U}_n$ be the set of $n$-dimensional random vectors that satisfy UVS,   i.e., 
$$\mathcal{L}^{\rm U}_n=\left\{(X_1,\dots,X_n)\in \mathcal L_n^0: \VaR_{p}\left( \sum_{i=1}^n X_{i}\right) \ge  \sum_{i=1}^n\VaR_{p}(X_i) \mbox{~for all~} p\in(0,1)\right\}.$$
The set $\mathcal{L}^{\rm U}_n$ is quite rich, as can be seen in Section \ref{sec:ind}. We only give one example here, which is the primary example to keep in mind.
\begin{example}\label{ex:pareto0}
 The cumulative distribution function of the Pareto distribution with tail parameter $\alpha\in(0,\infty)$, denoted by ${\rm Pareto}(\alpha)$, is given by $$F(x)=1-\left(\frac{1}{x}\right)^\alpha, ~~~ x \ge 1.$$ If $\alpha\le 1$, the Pareto distribution has infinite mean. Theorem 1 of \cite{CEW25} implies that $(X_1,\dots,X_n)$ satisfies strict  UVS for independent $X_1,\dots,X_n\sim {\rm Pareto}(\alpha)$ with $\alpha\in(0,1]$.   In Example \ref{ex:pareto}, we will see that the same result holds for random vectors with independent but not necessarily identically distributed infinite-mean Pareto random variables. As an illustration, Figure \ref{fig:pareto} plots $\VaR_p(X_1+X_2+X_3)$ and $\VaR_p(X_1)+\VaR_p(X_2)+\VaR_p(X_3)$ for $p\in(0,0.99)$ by simulations, where $X_1\sim {\rm Pareto(0.7)}$, $X_2\sim {\rm Pareto(0.8)}$, and $X_3\sim {\rm Pareto(0.9)}$. The results for $p\in(0,0.95)$ and $(0.95,0.99)$ are plotted separately for visibility. 

 \begin{figure}[h]
\centering
\includegraphics[width=15cm]{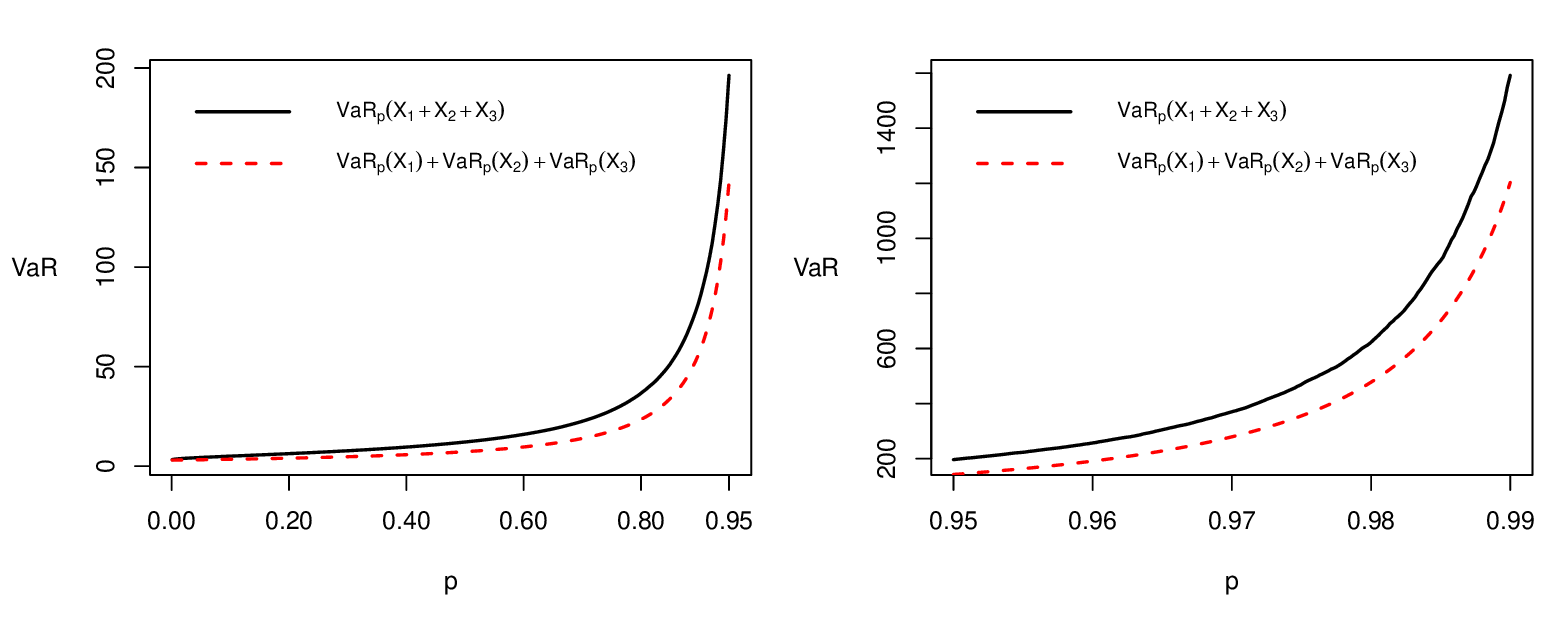}
\caption{$\VaR_p(X_1+X_2+X_3)$ and $\VaR_p(X_1)+\VaR_p(X_2)+\VaR_p(X_3)$ for $p\in(0,0.99)$, where $X_1\sim {\rm Pareto(0.7)}$, $X_2\sim {\rm Pareto(0.8)}$, and $X_3\sim {\rm Pareto(0.9)}$.}    \label{fig:pareto}
\end{figure}
 \end{example}

\begin{remark}   UVS requires a continuum of VaR inequalities, and is  thus different from the asymptotic superadditivity of VaR in the literature (e.g., \citealp{ELW09,mainik2010optimal}), 
that is often defined as 
$$\liminf_{p\uparrow 1}\frac{\VaR_p(X_1+\dots+X_n)}{\VaR_p(X_1)+\dots+\VaR_p(X_n)}> 1.$$
\end{remark}

Our paper is the first to systematically study the set 
$\mathcal L_n^{\rm U}$, as well as the related set   $\mathcal{L}^{\rm WU}_n$ 
defined by \begin{equation*}  
\mathcal{L}^{\rm WU}_n=\left\{(X_1,\dots,X_n) \in \mathcal L^0_n: (\theta_1 X_1,\dots,\theta_n X_n)\in \mathcal L^{\rm U}_n \mbox{ for all }(\theta_1,\dots,\theta_n)\in\Delta_n \right\},\end{equation*}
where $\Delta_n=\{(\theta_1,\dots,\theta_n)\in [0,1]^n:  \sum_{i=1}^n \theta_i=1\}$. 
Equivalently, elements of $\mathcal{L}^{\rm WU}_n$   satisfy \emph{weighted universal VaR superadditivity (WUVS)}, that is,
\begin{equation}
    \label{eq:main2}\VaR_{p}\left( \sum_{i=1}^n \theta_iX_{i}\right) \ge  \sum_{i=1}^n\VaR_{p}(\theta_iX_i) \mbox{~for all~} p\in(0,1) \mbox{~and~} (\theta_1,\dots,\theta_n) \in\Delta_n.
\end{equation} 
\emph{Strict WUVS} is defined by requiring \eqref{eq:main2}  to hold  with strict inequality for all $p\in(0,1)$ and all $(\theta_1,\dots,\theta_n)\in \Delta_n$ with at least two positive components. Strict WUVS has stronger implications in decision making than WUVS, which will be discussed in Proposition \ref{prop:rm}.

The random vector $(X_1,\dots,X_n)$ with  iid  Pareto$(\alpha)$ components   for $\alpha\le 1$ in Example \ref{ex:pareto0} satisfies strict  WUVS by Theorem 1 of \cite{CEW25}. 
By taking $\theta_i=1/n$ for all $i \in [n]$, we see that WUVS is stronger than UVS, and thus the inclusion $$\mathcal{L}^{\rm WU}_n \subseteq \mathcal{L}^{\rm U}_n$$ is clear.
The following example shows that the set $\mathcal{L}^{\rm U}_2$ is strictly larger than $\mathcal{L}^{\rm WU}_2 $.  
\begin{example}
\label{ex:stp}
The St.~Petersburg lottery pays $2^k$ with probability $2^{-k}$ for each $k \in \N$.
Let
$X_1, X_2$ be independent St.~Petersburg lotteries. 
% and $X^c_1, X^c_2$ be comonotonic St.~Petersburg lotteries.
 \citet[Proposition 3]{CHWZ25} yields that $(X_1, X_2)$ satisfies UVS;  {see the top panels in Figure \ref{fig:st},  where $\VaR_p(X_1+X_2)$ and $\VaR_p(X_1)+\VaR_p(X_2)$ over $p\in(0,0.99)$ are plotted.}
Let $\theta_1=0.1$ and $\theta_2=0.9$. %It is clear that $\p(\theta_1X^c_1+\theta_2X^c_2>1.9)=\p(X>1.9)=\p(X>1)={1}/{2}.$
We can check 
\begin{align*}\p(\theta_1 X_1+\theta_2 X_2>3.8)&=\p(X_2\ge 8)+\p(X_2=4)\p(X_1\ge 4)+\p(X_2=2)\p(X_1\ge 32)\\
&=\frac{13}{32} \le 0.45,  %\frac{1}{2},
% =\p(\theta_1 X_1^c+\theta_2 X_2^c >3.8),
\end{align*}
 which implies 
$\VaR_{0.55}(\theta_1 X_1+\theta_2 X_2)\le 3.8$. Moreover, $$\VaR_{0.55}(\theta_1 X_1)+\VaR_{0.55}(\theta_2 X_2)=\VaR_{0.55}( X_1)=4.$$
% $\theta_1 X_1^c+\theta_2 X_2^c \not\le_{\rm st} \theta_1 X_1+\theta_2 X_2$ 
 Hence, $(X_1, X_2)$ does not satisfy WUVS;  {see the bottom panels in Figure \ref{fig:st},  where $\VaR_p(0.1X_1+0.9X_2)$ and $\VaR_p(0.1X_1)+\VaR_p(0.9X_2)$ over $p\in(0,0.99)$ are plotted.}  
\begin{figure}[h]
\centering
\includegraphics[width=15cm]{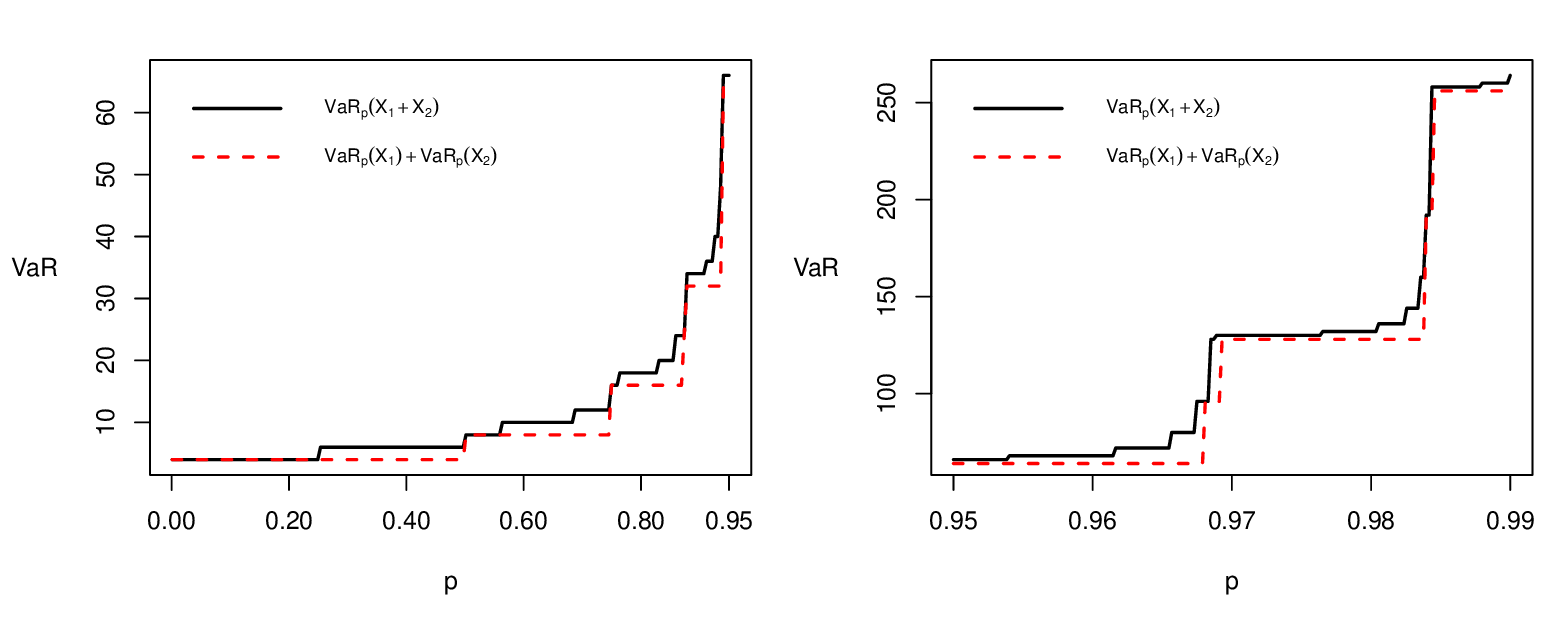}
\centering
\includegraphics[width=15cm]{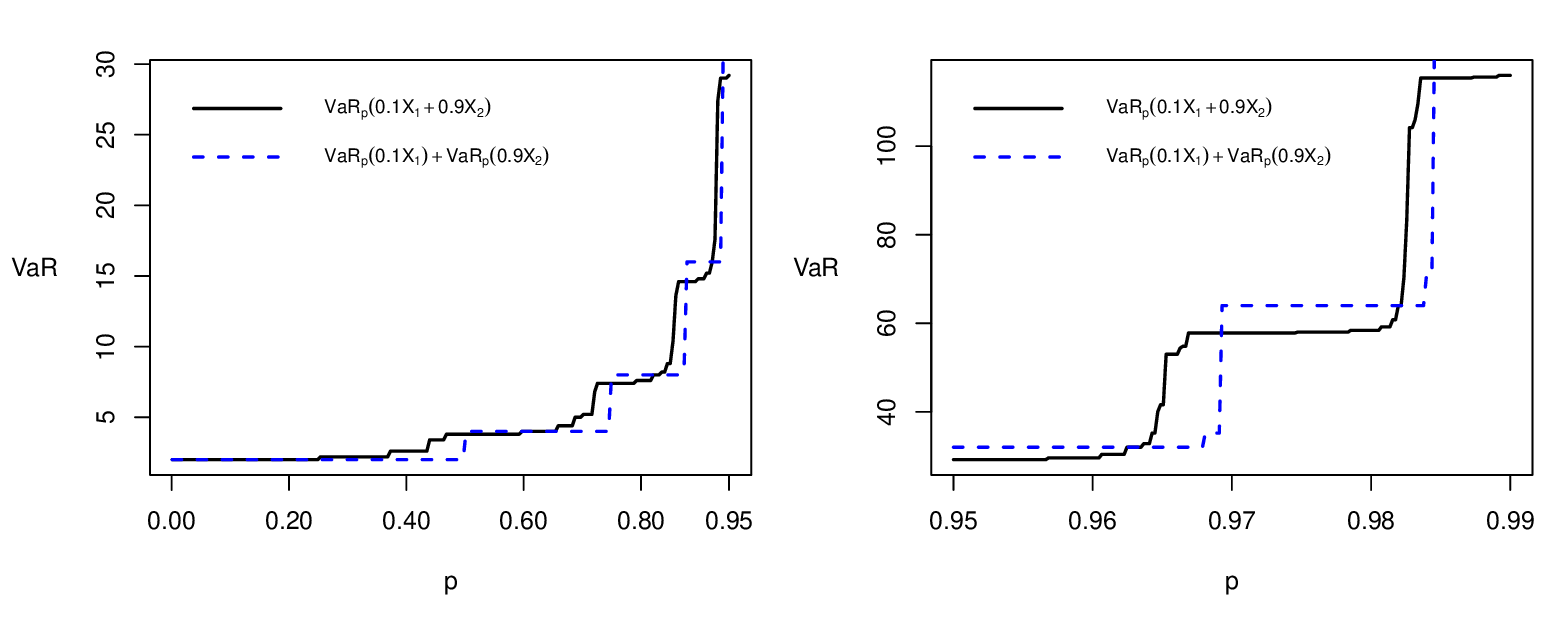}
\caption{Top panels: $\VaR_p(X_1+X_2)$ and $\VaR_p(X_1)+\VaR_p(X_2)$ for $p\in(0,0.99)$, where $X_1$ and $X_2$ are iid St.~Petersburg lotteries. Bottom panels: $\VaR_p(0.1X_1+0.9X_2)$ and $\VaR_p(0.1X_1)+\VaR_p(0.9X_2)$ for $p\in(0,0.99)$, where $X_1$ and $X_2$ are iid St.~Petersburg lotteries.}    \label{fig:st}
\end{figure}
\end{example}
Similarly, the strict inclusion $\mathcal{L}^{\rm WU}_n \subsetneq \mathcal{L}^{\rm U}_n$
holds for all $n\ge 2$, which can be easily seen by adding constant components to the random vector in Example \ref{ex:stp}.

\section{General properties of UVS and WUVS} \label{sec:VaR}

\subsection{Equivalent conditions}
We present some first observations on UVS and WUVS. Random variables $X_1,\dots,X_n$ are \emph{comonotonic} if there exist increasing functions $f_1,\dots,f_n$ and a random variable $Z$ such that $(X_1,\dots,X_n)=(f_1(Z),\dots,f_n(Z))$. In other words, $X_1,\dots,X_n$ are perfectly positively dependent. Here, $f_1,\dots,f_n$ can be chosen as the left inverse functions of the cumulative distribution functions of $X_1,\dots,X_n$, and $Z$ can be taken as a uniform random variable on $(0,1)$. If $X_1^c,\dots,X_n^c$ are comonotonic and have the same marginals as $X_1,\dots,X_n$, we call $(X_1^c, \dots, X_n^c)$ a \emph{comonotonic copy} of $(X_1, \dots, X_n)$.  
\begin{proposition} \label{prop:como}  
Suppose that  $\E[(X_i)_+]< \infty$ for all $i\in [n]$. The following are equivalent:
\begin{enumerate}[(i)]
    \item $(X_1, \dots, X_n)$ satisfies UVS;
    \item  $(X_1, \dots, X_n)$ satisfies WUVS;
        \item  $(X_1, \dots, X_n)$ satisfies universal VaR additivity;
    \item $X_1,\dots,X_n$ are comonotonic.
\end{enumerate}
\end{proposition}
\begin{proof}
The directions (iv)$\Rightarrow$(iii)$\Rightarrow$(i)
and 
(iv)$\Rightarrow$(ii)$\Rightarrow$(i) are clear. 
 It remains to show the (i)$\Rightarrow$(iv) direction. The proof follows the arguments in the proof of Proposition 2 of \cite{CEW25}. Assume $(X_1, \dots, X_n)\in \mathcal{L}^{\rm U}_n$. For a random variable $X$ and $p\in (0,1)$,  the ES at level $p$ is defined as
$$\ES_p(X)=\frac{1}{1-p}\int_p^1 \VaR_u(X) \d u.$$
Since $\E[(X_i)_+]< \infty$ for all $i\in [n]$, we have $\ES_p(X_i)< \infty$ for all $i\in[n]$ and $p \in (0,1)$. 
UVS  and subadditivity of ES imply
$$\sum_{i=1}^n \ES_p(X_i)\le  \ES_p\left(\sum_{i=1}^n X_i\right)\mbox{ and }\sum_{i=1}^n \ES_p(X_i)\ge  \ES_p\left(\sum_{i=1}^n X_i\right).$$
Therefore, $\ES_p(\sum_{i=1}^n X_i)=\sum_{i=1}^n \ES_p(X_i)$. By \citet[Theorem 5]{WZ21}, $(X_1, \dots, X_n)$ is $p$-concentrated for all $p\in (0,1)$; see \cite{WZ21} for the definition of $p$-concentration. \citet[Theorem 4]{WZ21} further implies that $X_1, \dots, X_n$ are comonotonic. Note that although Theorem 5 of \cite{WZ21} is stated for random variables with finite mean, its proof only needs upper integrability and hence it can be applied here.
\end{proof}
A direct consequence of Proposition \ref{prop:como} is that  UVS (and WUVS) cannot hold for independent and non-degenerate random vectors where each component is integrable.
Another consequence is that strict UVS cannot hold for any random vectors with all integrable components.
In Sections \ref{sec:cs}--\ref{sec:invsub}, we study several classes of risks such that strict WUVS holds for random vectors with components drawn from these
classes. These risks are non-degenerate and thus are extremely heavy-tailed in the sense that their first absolute moments are infinite.

 \begin{remark}
By a similar proof of Proposition \ref{prop:como}, we can show that for random variables $X_1,\dots,X_n$ with $\E[(X_i)_-]< \infty$ for all $i\in [n]$:  $(X_1,\dots, X_n)$ satisfies universal VaR subadditivity if and only if $X_1,\dots,X_n$ are comonotonic. 
% \begin{equation}\label{eq:sub}\mbox{universal VaR subadditivity holds for $X_1,\dots,X_n$ if and only if $X_1,\dots,X_n$ are comonotonic.}\end{equation} 
This result generalizes Theorem 1 of \cite{IK25}, where $X_1,\dots,X_n$ are assumed to be integrable.  Theorem 2.2 of \cite{M25} provides the same ``if and only if'' statement with the assumption that $X_1,\dots,X_n$ are bounded from below. Note that for any random variable $X$ that is bounded from below, $\E[(X)_-]<\infty$.
\end{remark}

Next, we present a useful lemma for our main results. We show stochastic dominance representations of UVS and WUVS, which also yield strong implications for decision making. For random variables $X$ and $Y$, $X$ is smaller than $Y$ in \emph{stochastic order}, denoted by $X\le_{\rm st}Y$, if $\p(X>t)\le \p(Y>t)$ for all $t\in\R$. 
We refer to \cite{MS02} and  \cite{SS07} for properties of stochastic order.

\begin{lemma}\label{lem:st}
% A random vector $(X_1, \dots, X_n)$ satisfies UVS if and only if
% $$X_1^c+\dots+X_n^c\le_{\rm st} X_1+\dots+X_n,$$
% where $(X_1^c,\dots,X_n^c)$ is a comonotonic copy of $(X_1,\dots,X_n)$.
Let $(X_1, \dots, X_n) \in \mathcal L_n^0$ and $(X_1^c,\dots,X_n^c)$ be a comonotonic copy of $(X_1,\dots,X_n)$. We have
\begin{enumerate}[(i)]
\item $(X_1, \dots, X_n)$ satisfies UVS if and only if  $X_1^c+\dots+X_n^c\le_{\rm st} X_1+\dots+X_n$; 
\item $(X_1, \dots, X_n)$ satisfies WUVS if and only if  $\theta_1X_1^c+\dots+\theta_nX_n^c\le_{\rm st} \theta_1X_1+\dots+\theta_nX_n$ for all $(\theta_1, \dots, \theta_n) \in \Delta_n$. 
\end{enumerate}
% {\color{red}If further $\sum_{i=1}^n X_i$ and $\sum_{i=1}^n X^{c}_i$ are continuous random variables, we have
% \begin{enumerate}
% \item[(iii)] $(X_1, \dots, X_n)$ satisfies strict UVS if $X_1^c+\dots+X_n^c<_{\rm st} X_1+\dots+X_n$; 
% \item[(iv)] $(X_1, \dots, X_n)$ satisfies strict WUVS if $\theta_1X_1^c+\dots+\theta_nX_n^c<_{\rm st} \theta_1X_1+\dots+\theta_nX_n$ for all $(\theta_1, \dots, \theta_n) \in \Delta_n$. 
% \end{enumerate}}
\end{lemma}
\begin{proof}
For (i), the result holds since for any random variables $X$ and $Y$, $X\le_{\rm st} Y$ if and only if $\VaR_p(X)\le\VaR_p(Y)$ for all $p\in(0,1)$ and 
\begin{equation*}\VaR_p\(\sum_{i=1}^nX_i^c\)=\sum_{i=1}^n\VaR_p\(X_i^c\)=\sum_{i=1}^n\VaR_p\(X_i\).\end{equation*}
For (ii), $(X_1, \dots, X_n) \in \mathcal{L}^{\rm WU}_n$ if and only if $(\theta_1 X_1, \dots, \theta_n X_n) \in \mathcal{L}^{\rm U}_n$ for all $(\theta_1, \dots, \theta_n) \in \Delta_n$. By (i), we have the result. 
\end{proof}

 \subsection{Convex transformation}

Strictly increasing and convex transformations provide a natural way to make risks more extreme in the upper tail. Indeed, such transformations make continuous random variables larger in convex transform order, i.e., more right-skewed in some sense; see, e.g., \cite{SS07} for a discussion on convex transform order. If a random vector $(X_1,\dots,X_n)$ with identically distributed components satisfies WUVS, then for any increasing and convex function $f$,  $(f(X_1),\dots,f(X_n))$ also satisfies WUVS. To see this,  for any $ (\theta_1,\dots,\theta_n)\in\Delta_n$, we have
$$ f(X_1)=f\(\sum_{i=1}^n\theta_i X_i^c\)\le_{\rm st}f\(\sum_{i=1}^n\theta_i X_i\)\le \sum_{i=1}^n\theta_i f(X_i),$$
where $(X_1^c,\dots,X_n^c)$ is a comonotonic copy of $(X_1,\dots,X_n)$.  
This observation has been used by \cite{CEW25}, \cite{Muller25}, and \cite{CS25}  to define several classes of distributions for random vectors in $\mathcal{L}^{\rm WU}_n$ with identically distributed components; see Section \ref{sec:invsub}. 

% We generalize this observation to non-identical convex transformations in the following theorem.  
% We show that a WUVS random vector can still be preserved in $\mathcal{L}^{\rm WU}_n$ if different strictly increasing and convex transformations are applied to its coordinates.
% if strictly increasing and convex transformations, possibly different, are applied to the marginal risks of an iid random vector in $\mathcal{L}^{\rm WU}_n$, the transformed random vector remains in 
% $\mathcal{L}^{\rm WU}_n$. 
Our first main result shows that  UVS and WUVS, as well as their strict versions, are preserved under increasing and convex transformations, but in two different senses: WUVS allows for different convex transforms, whereas UVS requires identical transforms and identical marginals.

% \begin{theorem}\label{pro:SAw}

% Suppose that $(X_1, \dots, X_n) \in \mathcal{L}^0_n$.
% \begin{enumerate}[(i)]
% \item If $(X_1, \dots, X_n)$ 
% satisfies WUVS, then $(f_1(X_1), \dots, f_n(X_n))$ satisfies WUVS for all increasing and convex functions  $f_1, \dots, f_n$. 
%  \item If $(X_1, \dots, X_n)$ 
% satisfies strict WUVS, then $(f_1(X_1), \dots, f_n(X_n))$ satisfies strict WUVS for all strictly increasing and convex functions  $f_1, \dots, f_n$.   
% \item If $(X_1, \dots, X_n)$ 
%     satisfies   WUVS and has positive density on $\prod_{i=1}^n[\essinf X_i,  \infty)$, then $(f_1(X_1), \dots, f_n(X_n))$ satisfies strict WUVS for all    increasing and strictly convex functions   $f_1, \dots, f_n$. 
%     \com{Use $\prod$ not $\Pi$}

% \end{enumerate}
% For the following statements, we further assume $(X_1, \dots, X_n)$ has identical marginals.  
% \begin{enumerate}
% \item[(iv)] If $(X_1, \dots, X_n)$ 
%     satisfies  UVS, then $(f(X_1), \dots, f(X_n))$ satisfies UVS for all    increasing and convex functions  $f$. 

% \item[(v)] If $(X_1, \dots, X_n)$ 
%     satisfies strict UVS, then $(f(X_1), \dots, f(X_n))$ satisfies strict UVS for all  strictly  increasing and convex functions  $f$. 
   
% \item[(vi)] If $(X_1, \dots, X_n)$ 
%     satisfies  UVS and has a positive density on $\prod_{i=1}^n [\essinf X_i, \infty)$, then we have $(f(X_1), \dots, f(X_n))$ satisfies strict UVS for all    increasing and strictly convex functions  $f$. 
% \end{enumerate} 
% \end{theorem} 

\begin{theorem}\label{th:convex}

For   $(X_1, \dots, X_n) \in \mathcal{L}^0_n$, we have the following statements.
\begin{enumerate}[(i)]
\item If $(X_1, \dots, X_n)$ 
satisfies (strict) WUVS, then $(f_1(X_1), \dots, f_n(X_n))$ satisfies (strict) WUVS for all (strictly) increasing and convex functions  $f_1, \dots, f_n$.  
\item[(ii)] If $(X_1, \dots, X_n)$ 
    satisfies (strict)  UVS and has  identical marginals, then $(f(X_1), \dots, f(X_n))$ satisfies (strict) UVS for any    (strictly)  increasing and convex function  $f$. 
\end{enumerate} 
\end{theorem}

The proof of Theorem~\ref{th:convex} relies on the following lemma, which may be of independent interest. In this lemma,   VaR superadditivity is only required (and proved) for a fixed $p$, and hence it is not universal.  
\begin{lemma} 
\label{lem:fixed-p-convex-transform}
Let $p\in(0,1)$ and   $(X_1, \dots, X_n) \in \mathcal{L}^0_n$, and suppose 
\begin{equation}
    \label{eq:fixed-p}
\operatorname{VaR}_p\left(\sum_{i=1}^n\theta_iX_i\right)
\ge \sum_{i=1}^n\theta_i\operatorname{VaR}_p(X_i)
\qquad\text{for all }(\theta_1,\dots,\theta_n)\in\Delta_n.
\end{equation}
Then, for any increasing convex functions
$f_1,\ldots,f_n:\mathbb R\to\mathbb R$,
\begin{equation}
    \label{eq:fixed-p2}
\operatorname{VaR}_p\left(\sum_{i=1}^n\theta_i f_i(X_i)\right)
\ge \sum_{i=1}^n\theta_i\operatorname{VaR}_p(f_i(X_i))
\qquad\text{for all }(\theta_1,\dots,\theta_n)\in\Delta_n.
\end{equation}
If the inequality in \eqref{eq:fixed-p} is strict whenever at least two $\theta_i$ are positive, and each $f_i$ is strictly increasing, then the inequality in \eqref{eq:fixed-p2} is also strict whenever at least two $\theta_i$ are positive.
\end{lemma}

\begin{proof}
Write $q_i=\operatorname{VaR}_p(X_i)$ for each $i$. By positive homogeneity,  \eqref{eq:fixed-p} can be extended from $\Delta_n$ to $[0,\infty)^n$. For each
$i\in[n]$, by convexity, we can  choose  $a_i\ge0$ such that
\[
f_i(x)\ge f_i(q_i)+a_i(x-q_i),\qquad x\in\mathbb R.
\]
For any $(\theta_1,\dots,\theta_n)\in\Delta_n$, monotonicity and cash additivity of VaR give
\[
\begin{aligned}
\operatorname{VaR}_p\left(\sum_{i=1}^n\theta_i f_i(X_i)\right)
&\ge \sum_{i=1}^n\theta_i f_i(q_i)
+\operatorname{VaR}_p\left(\sum_{i=1}^n\theta_i a_iX_i\right)
-\sum_{i=1}^n\theta_i a_iq_i\\
&\ge \sum_{i=1}^n\theta_i f_i(\VaR_{p}(X_i))
=\sum_{i=1}^n\theta_i\operatorname{VaR}_p(f_i(X_i)),
\end{aligned}
\]
where the second inequality applies \eqref{eq:fixed-p} with $(\theta_i a_i)_{i\in [n]}$ and the last equality follows from the ordinal invariance of the quantile; see e.g., \citet[Theorem 1]{FLW23}. 
If each $f_i$ is strictly increasing, convexity ensures $a_i>0$.
Hence $(\theta_1a_1,\dots,\theta_na_n)$ has at least two positive components whenever
$(\theta_1,\dots,\theta_n)$ does, and the second inequality is strict under the strict
hypothesis. This proves the strictness statement.
\end{proof}

\begin{proof}[Proof of Theorem~\ref{th:convex}]
\begin{enumerate}[(i)]

%\item If $F$ is a continous distribution, we have $\p(X_1=\dots=X_n)=0$. Since $f$ is increasing and strictly convex, the second inequality in \eqref{eq:UVS} is strict almost surely; hence, $(f(X_1), \dots, f(X_n)$ satisfies strict UVS.
%\item If $f$ is strictly increasing and $(X_1,\dots,X_n)$ satisfies strict UVS, the first inequality in \eqref{eq:UVS} is strict in the sense that $$\p(f(X)>x)<\p\(f\left(\frac{1}{n} X_1+ \dots+ \frac{1}{n} X_n\right)>x\) \mbox{for all $x\in\R$}.$$ Since each $X_i$ is continuous, we have $\sum_{i=1}^n f(X_i)$ and $\sum_{i=1}^n f(X)$ are continuous random variables as $f$ is continuous.  By Lemma \ref{lem:st}, $(f(X_1), \dots, f(X_n))$ satisfies strict UVS.
\item 
 If $(X_1, \dots, X_n)$ satisfies WUVS, we have 
\eqref{eq:fixed-p} holds for all $p\in (0,1)$. By Lemma \ref{lem:fixed-p-convex-transform}, we have 
\eqref{eq:fixed-p2} holds for all $p \in (0,1)$ and any increasing and convex functions $f_1, \dots, f_n$. Hence, $(f_1(X_1), \dots, f_n(X_n))$ satisfies WUVS. 
If $(X_1, \dots, X_n)$ satisfies strict WUVS, the inequality in \eqref{eq:fixed-p} is strict for all $p\in (0,1)$ whenever at least two $\theta_i$ are positive. By Lemma \ref{lem:fixed-p-convex-transform}, we have that the inequality in 
\eqref{eq:fixed-p2} is strict for all $p \in (0,1)$ and any strictly increasing and convex functions $f_1, \dots, f_n$ whenever at least two $\theta_i$ are positive. Hence, $(f_1(X_1), \dots, f_n(X_n))$ satisfies strict WUVS.

\item Since $f$ is convex, we have
\begin{equation}\label{eq:UVS1}
f\left(\frac{1}{n} X_1+ \dots+ \frac{1}{n} X_n\right)\le \frac{1}{n}f(X_1)+\dots+\frac{1}{n}f(X_n),
\end{equation} 
Since $\VaR$ is monotone,
$
\VaR_p\left(f\left( X_1/n+ \dots+  X_n/n\right)\right)
\le \VaR_p\left(f(X_1)/n+\dots+f(X_n)/n\right).
$
Since $f$ is increasing, we have $f\left(\VaR_p\left( X_1/n+ \dots+ X_n/n\right)\right)=\VaR_p\left(f\left( X_1/n+ \dots+X_n/n\right)\right)$.
By  $(X_1, \dots, X_n) \in \mathcal{L}^{\rm U}$, we further have
\begin{align}
\VaR_p(f(X))=f(\VaR_p(X))&\le f\left(\VaR_p\left(\frac{1}{n}X_1 + \dots+ \frac{1}{n}X_n \right)\right)\label{eq:UVS3}\\
&\le \VaR_p\left(\frac{1}{n}f(X_1)+\dots+\frac{1}{n}f(X_n)\right)\nonumber.
\end{align}
% Assume $(X_1, \dots, X_n) \in \mathcal{L}^{\rm U}_n$. We have
% $$X\le_{\rm st}\frac{1}{n} X_1+ \dots+ \frac{1}{n} X_n.$$
% As $f$ is increasing and  convex, we have 
% \begin{equation}\label{eq:UVS}
% f(X)\le_{\rm st}f\left(\frac{1}{n} X_1+ \dots+ \frac{1}{n} X_n\right)\le \frac{1}{n}f(X_1)+\dots+\frac{1}{n}f(X_n),
% \end{equation}
% which implies $(f(X_1), \dots, f(X_n)) \in \mathcal{L}^{\rm U}_n$.

If $f$ is strictly increasing and $(X_1, \dots, X_n)$ satisfies strict UVS, we have that the  inequality in \eqref{eq:UVS3} is strict for all $p\in (0,1)$.
Hence, $(f(X_1), \dots, f(X_n))$ satisfies strict UVS.\qedhere
%As $f$ is strictly increasing and $(X_1, \dots, X_n)$ is strict UVS, the $\le_{\rm st}$ in \eqref{eq:UVS} is $<_{\rm st}$. Hence, $(f(X_1), \dots, f(X_n))$ is strict UVS by Lemma \ref{lem:con}.
\end{enumerate}
\end{proof}

  \begin{remark}\label{rem:strict}
 A weaker notion of strict UVS (resp.~WUVS) can be defined by replacing ``all $p\in (0,1)$" with ``a.e.~$p\in (0,1)$" in the definition of strict UVS (resp.~WUVS).
A random vector $(X_1, \dots, X_n)$ satisfies this notion of strict UVS (called ``a.e.~strict UVS")  if and only if  there exists $Y\laweq \sum_{i=1}^n X_i$ and $Z\laweq \sum_{i=1}^n X^c_i$ such that $\p(Y>Z)=1$.  If we replace ``strict" with  ``a.e.~strict" in Theorem \ref{th:convex}, the statements still hold.  %If we only require that $(X_1, \dots, X_n)$ has density in (vi), then we have $(f(X_1), \dots, f(X_n))$ satisfies ``a.e. strict UVS". The proof is straightforward as
%  $\p(X_1=\dots=X_n)=0$ when density exists. Hence, we have the inequality in \eqref{eq:UVS1} is strict almost surely and the inequality in \eqref{eq:UVS2} is strict for a.e. $p\in (0,1)$. 
 \end{remark}

The following example shows that the assumption of identical marginals cannot be removed for statement (ii) in Theorem \ref{th:convex}.
%\com{Revised the sentence: we cannot say it is necessary, because it is not.}
\begin{example}
Let $X$ and $Y$ be iid St.~Petersburg lotteries and let $Z=Y+10$. As discussed in Example \ref{ex:stp}, $(X,Y)$ satisfies UVS. As VaR is translation invariant, $(X,Z)$ also satisfies UVS. Applying the convex transform $f(x)=x^2$, $x\ge 0$, to $X$ and $Z$, we have %\com{I think $Z$ should be $Y$ in the right; revised}
\begin{align*}
\p\left(X^2+(Y+10)^2\le 208\right)
&\ge
\p \left((X,Y)\in \{(2,2),(4,2),(2,4),(8,2)\} \right) \\
%\p(X=2,Z=2)
%+\p(X=4,Z=2)+\p(X=2,Z=4%)  \\
%&\quad
%+\p(X=8,Z=2) \\
&=
\frac14+\frac18+\frac18+\frac1{16}  =
\frac9{16}
>0.55.
\end{align*}
Hence, $\VaR_{0.55}(f(X)+f(Z))\le 208< 212=\VaR_{0.55}(f(X))+\VaR_{0.55}(f(Z))$.
\end{example}

\begin{example}\label{ex:pareto}
Consider the Pareto($\alpha$) distribution in Example \ref{ex:pareto0}. 
  If $\alpha\in(0,1]$, a Pareto($\alpha$)  random variable can be obtained through a  convex transform $x\mapsto x^{1/\alpha}$ from a Pareto($1$) random variable. Note that an iid random vector with Pareto$(1)$ components satisfies strict WUVS  \citep[Theorem 1]{CEW25}. Then,  by Theorem \ref{th:convex}, $(X_1,\dots,X_n)$ satisfies {strict} WUVS if $X_i\sim {\rm Pareto}(\alpha_i)$ with  $\alpha_i\in(0,1]$ for all $i\in[n]$ and $X_1,\dots,X_n$ are independent. 
  This result, as a conjecture explicitly mentioned by \citet[Section 6.1]{CEW25b}, was first proved by \citet[Theorem 3.5]{M25} in the non-strict sense. We recover it here as a consequence of Theorem \ref{th:convex}. 
  We further investigate the application of Theorem \ref{th:convex}   in Section \ref{sec:invsub}.
 \end{example}

 \citet[Proposition 3.16]{M25} considered a subclass in $\mathcal L_n^{\rm U}$ that is closed under convex transforms, which is covered by Theorem \ref{th:convex}.

The strictness of UVS and WUVS is quite a strong requirement, because we need the inequalities in \eqref{eq:main}--\eqref{eq:main2} to be strict for all $p$. 
Theorem \ref{th:convex} assumes strict UVS or WUVS  for  $(X_1,\dots,X_n)$ to obtain the strict UVS or WUVS  after convex transformation. 
In the next result, we show that this is possible without assuming strict UVS or WUVS: we can instead require strictly convex transformations to obtain the strict versions of these properties, under an additional assumption on the density of the random vector. The proof of this result requires different arguments than those in the proof of Theorem \ref{th:convex}.

%\com{added the bridging paragraph. Is the last sentence correct?}

\begin{theorem}\label{pro:SAw}

Suppose that $(X_1, \dots, X_n) \in \mathcal{L}^0_n$  has positive density on $\prod_{i=1}^n (\essinf X_i,  \infty)$ and is atomless. 
%\com{What happens if $\essinf X_i=-\infty$?}
%\com{LY: I changed it to open range}
\begin{enumerate}[(i)]
\item If $(X_1, \dots, X_n)$ 
    satisfies   WUVS, then $(f_1(X_1), \dots, f_n(X_n))$ satisfies strict WUVS for all    increasing and strictly convex functions   $f_1, \dots, f_n$. 
 \item    If $(X_1, \dots, X_n)$ 
    satisfies  UVS and has identical marginals, then $(f(X_1), \dots, f(X_n))$ satisfies strict UVS for any    increasing and strictly convex function  $f$. 
\end{enumerate}

\end{theorem} 

\begin{proof}
% We follow the arguments in \eqref{eq:convex}, \eqref{eq:WUVS1} and \eqref{eq:UVS1} in the proof of Theorem \ref{th:convex} with the following modifications.
%  \com{I don't understand what this means. Are these two proofs similar or not? I emphasized their difference above.}
%  \com{LY: I don't want to repeat the argument around \eqref{eq:convex}, \eqref{eq:WUVS1} and \eqref{eq:UVS1} here again. RW: Let's emphasize the difference then.}
\begin{enumerate}[(i)]

\item   %Since $(X_1, \dots, X_n)$ has positive density on $\prod_{i=1}^n [\essinf X,  \infty)$, we have $$\theta_if_i(X_i)-\sum_{i=1}^n\theta_i \VaR_p(f_i(X_i))> \sum_{i=1}^n \theta_ia_i(X_i-\VaR_p(X_i))~~~\mbox{a.s.}$$
Since $f_1, \dots, f_n$ are increasing and strictly convex, we can choose $a_1, \dots, a_n>0$ such that 
\begin{equation}\label{eq:convex}f_i(X_i)\ge f_i(\VaR_p(X_i))+a_i(X_i-\VaR_p(X_i)),~~~i \in [n].
 \end{equation}
Hence,
\begin{equation}\label{eq:WUVS1}
\left\{\sum_{i=1}^n \theta_i f_i(X_i) - \sum_{i=1}^n\theta_i \VaR_p(f_i(X_i))\le  \epsilon\right\}\subseteq \left\{ \sum_{i=1}^n \theta_i a_i X_i -\sum_{i=1}^n \theta_i a_i \VaR_p(X_i)\le \epsilon \right\}.
\end{equation}
 The key is to show that the difference of two sets in \eqref{eq:WUVS1} has a positive probability.
%Since $f_1, \dots, f_n$ are strictly convex, we have $a_1, \dots, a_n>0$ in \eqref{eq:convex}.
Without loss of generality, let $(\theta_1, \dots, \theta_n) \in \Delta_n$ with  $\theta_1>0$ and $\theta_2>0$.
Let $t,s>0$ such that $\theta_1a_1t-\theta_2a_2s=0$, $x_1=\VaR_p(X_1)+t$, $x_2=\VaR_p(X_2)-s$, and $x_i=\VaR_p(X_i)$ for $i \in [n]\setminus[2]$. Since $(X_1, \dots, X_n)$ is atomless, we can find $s$ small enough such that $x_2=\VaR_p(X_2)-s>\essinf(X_2)$.
For such $(x_1, \dots, x_n)$, we have $\sum_{i=1}^n \theta_i a_i (x_i-\VaR_p(X_i))=0$. By strict convexity of $f_1, \dots, f_n$ and \eqref{eq:convex},  we further have
$\sum_{i=1}^n \theta_i f_i(x_i)>\sum_{i=1}^n\theta_if_i(\VaR_p(X_i))+\theta_1a_1t-\theta_2a_2s=\sum_{i=1}^n\theta_if_i(\VaR_p(X_i))$. Hence, there exists $\delta>0$ small enough such that $\sum_{i=1}^n \theta_i f_i(x_i)-\sum_{i=1}^n\theta_if_i(\VaR_p(X_i))>\delta$. Let $\epsilon>0$ and $(x'_1, \dots, x'_n)=(x_1, \dots, x_n)-(0,\epsilon, 0, \dots, 0)$. For $(x'_1, \dots, x'_n)$, we have $\sum_{i=1}^n \theta_i a_i (x'_i-\VaR_p(X_i))=-\theta_2a_2\epsilon<0$.
It is clear that $f_1, \dots, f_n$ are continuous. Therefore, we can take $\epsilon$ small such that $\sum_{i=1}^n \theta_i f_i(x'_i)-\sum_{i=1}^n\theta_if_i(\VaR_p(X_i))>\delta'>0$ for some $\delta'<\delta$. Further, by the continuity of $f_1, \dots, f_n$, we can find an open area $O \subseteq \prod_{i=1}^n (\essinf X_i, \infty)$ such that 
$$\sum_{i=1}^n \theta_i f_i(y_i)-\sum_{i=1}^n \theta_i \VaR_p(f_i(X_i))>0 ~~~\mbox{and} ~~~\sum_{i=1}^n \theta_i a_i(y_i-\VaR_p(X_i))< 0$$
for all $(y_1, \dots, y_n)\in O$. Since
$(X_1, \dots, X_n)$ has positive density on $\prod_{i=1}^n (\essinf X_i, \infty)$, we have $\p((X_1, \dots, X_n) \in O)>0$. Moreover,  $\sum_{i=1}^n \theta'_iX_i$ is atomless.
Therefore, by \eqref{eq:WUVS1}, we have
$$\p\left(\sum_{i=1}^n \theta_i f_i(X_i) \le \sum_{i=1}^n\theta_i \VaR_p(f_i(X_i))\right)< \p\left(\sum_{i=1}^n \theta'_i X_i \le \sum_{i=1}^n \theta'_i \VaR_p(X_i)\right)\le p.$$
By the definition of $\VaR$, we further have  $\VaR_p(\sum_{i=1}^n \theta_i f_i(X_i) )> \sum_{i=1}^n\theta_i \VaR_p(f_i(X_i))$.
\item 
 %If each $X_i$ has a continuous and strictly increasing distribution function, 
% Since $(X_1, \dots, X_n)$ has density, we have $\p(X_1=\dots=X_n)=0$. Since $f$ is increasing and strictly convex, the inequality in \eqref{eq:UVS1} is strict almost surely.\com{This can only hold for a.e. $p$.}
% Following the argument in (i), we have $(f(X_1), \dots, f(X_n))$ is strictly UVS.
 It is clear that if $f$ is increasing and strictly convex, then $f$ is strictly increasing and continuous. We recall \eqref{eq:UVS1} in the proof of Theorem \ref{th:convex}, which implies 
 $$\left\{\frac{1}{n}f(X_1)+\dots+\frac{1}{n}f(X_n)\le \VaR_p(f(X))\right\} \subseteq \left\{\frac{1}{n}X_1+\dots+\frac{1}{n}X_n\le \VaR_p(X)\right\}$$
 where $X \laweq X_i$ for $i \in [n]$.
 Since $(X_1, \dots, X_n)$ is atomless, we can find $t>0$ small enough such that $\VaR_p(X)-t>\essinf(X)$.
Since $f$ is strictly convex, by taking $x_1=\VaR_p(X)+t$, $x_2=\VaR_p(X)-t$, $x_i=\VaR_p(X)$ for $i\in [n]\setminus[2]$ and some $t>0$, we have $\sum_{i=1}^n x_i/n=\VaR_p(X)$ and $f(x_1)/n+\dots+f(x_n)/n>f(x_1/n+\dots+x_n/n)=f(\VaR_p(X))=\VaR_p(f(X)).$ Following the proof in (i),
we can always find an open set $O\subseteq \prod_{i=1}^n (\essinf X_i, \infty)$ such that $f(y_1)/n+\dots+f(y_n)/n> \VaR_p(f(X))$ but $y_1/n+\dots+y_n/n<
\VaR_p(X)$ for all $(y_1, \dots, y_n) \in O$. Since $(X_1, \dots, X_n)$ has positive density on $\prod_{i=1}^n (\essinf X_i, \infty)$, we have $\p((X_1, \dots, X_n) \in O)>0$. Moreover, $\sum_{i=1}^n X_i$ is atomless. Hence, we have 
$$\p\left(\frac{1}{n}f(X_1)+\dots+\frac{1}{n}f(X_n)\le \VaR_p(f(X))\right)<\p\left(\frac{1}{n}X_1+\dots+\frac{1}{n}X_n\le \VaR_p(X)\right)\le p,$$ which gives strict UVS.\qedhere
\end{enumerate}
\end{proof}

 Theorem \ref{pro:SAw} will be applied to an iid random vector $(X_1,\dots,X_n)\in \mathcal{L}^{\rm WU}_n$  in Section \ref{sec:invsub} to obtain strict WUVS random vectors with non-identical marginals.
\begin{remark}
If we weaken the condition of having positive density on $\prod_{i=1}^n (\essinf X_i,  \infty)$ for statement (ii) of Theorem \ref{pro:SAw} to having a density, then $(f(X_1), \dots, f(X_n))$ satisfies ``a.e. strict" UVS as defined in Remark \ref{rem:strict}. The proof is quite straightforward, as the inequality in \eqref{eq:UVS1} holds strictly almost surely. Hence, we have $$
\VaR_p\left(f\left( \frac {X_1}{n}+ \dots+  \frac {X_n}{n}\right)\right)
<\VaR_p\left(\frac 1n f(X_1)+\dots+\frac 1n f(X_n)\right), \mbox{~~~a.e.~$p \in (0,1)$.}
$$   Following the proof of (ii) in Theorem \ref{th:convex}, we have the ``a.e.~strict" UVS of $(f(X_1), \dots, f(X_n))$.
\end{remark}
\subsection{Closure properties}
We next present some closure properties of  $\mathcal L^{\mathrm U}_n$
    and   $\mathcal L^{\mathrm{WU}}_n$. 
    A \emph{distributional mixture} of $\mathbf X$ and $\mathbf Y$ is another random vector $\mathbf Z$ whose distribution is given by a mixture of the distributions of $\mathbf X$ and $\mathbf Y$, that is, 
    $\mathbf Z\laweq \mathbf X\id_A+\mathbf Y\id_{A^c}$
for some event $A$ independent of $(\mathbf X,\mathbf Y)$.  
 As we will see from the next result, although the sets  $\mathcal L^{\mathrm U}_n$
    and   $\mathcal L^{\mathrm{WU}}_n$
    are not closed under distributional mixtures,
    two subsets of theirs
    are closed under distributional mixtures.
    Let $\mathcal{L}^{\rm IM}$ be the set of all random vectors with identically distributed components (with arbitrary dimension and dependence structure).  

\begin{theorem}
\label{th:closure}
    The sets $\mathcal L^{\mathrm U}_n$
    and   $\mathcal L^{\mathrm{WU}}_n$
    are closed under convergence in distribution. They are not closed under distributional mixture or convex combinations. 
The sets  $\mathcal L^{\mathrm{U}}_n \cap \mathcal{L}^{\rm IM}$ and  $\mathcal L^{\mathrm{WU}}_n \cap \mathcal{L}^{\rm IM}$ are closed under distributional mixture, but not closed under convex combinations.
\end{theorem}
\begin{proof}
The first statement follows from a few facts: First, the (usual) stochastic order is continuous in convergence in distribution.
Second, convergence in distribution of random vectors implies 
the convergence in distribution of their sum, as well as their marginal distributions.
Third, convergence in distribution of the marginal distributions implies the 
convergence in distribution of the comonotonic copy.
Putting all these facts together and using 
Lemma \ref{lem:st}, we arrive at the closedness under convergence in distribution.

We provide two counter-examples for the second statement in the case $n=2$ (the general case is analogous). %Since $\mathcal{L}^{\rm WU}_n \subseteq \mathcal{L}^{\rm U}_n$, it is sufficient to show the counter-examples for  $\mathcal{L}^{\rm WU}_n$. 
 Let $X$ and $Y$ be two non-degenerate independent random variables such that $\E[(X)_+]< \infty$, $\E[(Y)_+]< \infty$, and  $c =\E[X]=\E[Y]\in\R$. 
 By Lemma \ref{lem:st}, we have $\mathbf X=(X,c) \in \mathcal{L}^{\rm WU}_2 \subseteq \mathcal{L}^{\rm U}_2$ and $\mathbf Y=(c, Y) \in \mathcal{L}^{\rm WU}_2$. 
 For any $\alpha \in (0,1)$, we have $\E[(\alpha  X+(1-\alpha )c)_+]< \infty$ and $\E[(\alpha  c+(1-\alpha )Y)_+]< \infty$. By Proposition \ref{prop:como}, it is clear that
  $\alpha  \mathbf X+(1-\alpha ) \mathbf Y\notin \mathcal{L}^{\rm U}_2$ as $\alpha X+(1-\alpha )c$ and $\alpha c+(1-\alpha )Y$ are not comonotonic. Hence,  $\alpha  \mathbf X+(1-\alpha ) \mathbf Y\notin \mathcal{L}^{\rm WU}_2$.
Therefore, neither $\mathcal{L}^{\rm WU}_2$ nor  $\mathcal{L}^{\rm U}_2$ is closed under convex combinations. 
 Let $A \subseteq \Omega$ be  independent of $\mathbf X$ and $\mathbf Y$ such that $\p(A)=\alpha \in (0,1)$. %The distribution of $\mathbf X\id_A+\mathbf Y\id_{A^c}$ follows the mixture of the distributions for $\mathbf X$ and $\mathbf Y$. 
By construction, $\mathbf X\id_A+\mathbf Y\id_{A^c}$ is not comonotonic. By Proposition \ref{prop:como},  we get $\mathbf X\id_A+\mathbf Y\id_{A^c} \notin \mathcal{L}^{\rm U}_2$, and hence neither $\mathcal{L}^{\rm WU}_2$ nor  $\mathcal{L}^{\rm U}_2$ is closed under distributional mixture. We can easily construct counter-examples for $n >2$ by adding constant components to $\mathbf X$ and $\mathbf Y$. 

Next, we show the last statement. 
Take $\mathbf X =(X_1, \dots, X_n)\in \mathcal L^{\rm WU}_n\cap \mathcal{L}^{\rm IM},$ $ \mathbf Y=(Y_1, \dots, Y_n) \in \mathcal L^{\rm WU}_n\cap \mathcal{L}^{\rm IM} $,  and an event $A$ independent of $(
\mathbf X,\mathbf Y)$.    
For $(\theta_1,\dots,\theta_n)\in\Delta_n$ and $x \in \R$, we have 
\begin{align*}
\p\(\sum_{i=1}^n\theta_i(X_i\id_A+Y_i\id_{A^c})>x\)&=\p(A)\p\(\sum_{i=1}^n\theta_iX_i>x\)+\p(A^c)\p\(\sum_{i=1}^n\theta_iY_i>x\)\\
&\ge \p(A)\p\(X_1>x)+\p(A^c)\p(Y_1>x\)\\
&=\p(X_1\id_A+Y_1\id_{A^c}>x),
\end{align*}
which implies $(\theta_1(X_1\id_A+Y_1\id_{A^c}), \dots, \theta_n(X_n\id_A+Y_n\id_{A^c})) \in \mathcal{L}^{\rm U}_n$ and hence $\mathbf X\id_A+\mathbf Y\id_{A^c} \in \mathcal{L}^{\rm WU}_n$.
Clearly,  $\mathbf X\id_A+\mathbf Y\id_{A^c} \in \mathcal{L}^{\rm IM}$, and hence  it is in $\mathcal L^{\rm WU}_n\cap \mathcal{L}^{\rm IM} $.
By taking $\theta_1=\dots=\theta_n=1/n$, we have the result for $\mathcal{L}^{\rm U}_n\cap \mathcal{L}^{\rm IM}$.

Finally, we note that a convex combination of elements in $\mathcal{L}^{\rm U}_n\cap \mathcal{L}^{\rm IM}$
does not necessarily have identical marginals, and therefore $\mathcal{L}^{\rm U}_n\cap \mathcal{L}^{\rm IM}$ is not convex. For instance, let $X,Y,Z$ be iid Pareto(1). Both $(X,Y)$ and $(X,Z)$ satisfy UVS but  a convex combination of them does not have identical marginals. The case for $\mathcal{L}^{\rm WU}_n\cap \mathcal{L}^{\rm IM}$ is similar.
\end{proof}
The last statement in Theorem \ref{th:closure}
implies that for random vectors $\mathbf X $ and $\mathbf Y $  
satisfying UVS (resp.~WUVS) and each having identically distributed components, we have that  $ \mathbf X \id_A + \mathbf Y\id_{A^c}$  satisfies UVS (resp.~WUVS),
where $A$ is any event independent of $(\mathbf X,\mathbf Y)$.
Note that even if $\mathbf X$ and $\mathbf Y$   have independent components, the components of  $\mathbf {X} \id_A+\mathbf {Y} \id_{A^c}$ are not necessarily independent. Some simple cases below illustrate this.
\begin{example}
Let $\mathbf X,\mathbf Y\in \mathcal{L}^{\rm IM}$ and $\mathbf Z$ be a  distributional mixture of $\mathbf X,\mathbf Y$.
\begin{enumerate}[(i)]
    \item If $\mathbf X,\mathbf Y$ have  degenerate  marginals,  then  $\mathbf Z$ takes values at two points $(a,\dots,a)$ and $(b,\dots,b)$ for some $a,b\in \R$. 
    Therefore, $\mathbf Z$ has comonotonic (and further, identical) components. 
    \item  If $\mathbf X$ and $\mathbf Y$ have   comonotonic components, then, because they  also have identical marginals, their components must be almost surely identical. Then $\mathbf Z$ also has almost surely identical components. 
    \item If $\mathbf X$ has iid Pareto$(1)$ components and $\mathbf Y$ has iid Pareto$(1/2)$ components, then $\mathbf Z$ has components that are marginally distributed as a mixture of Pareto$(1)$ and Pareto$(1/2)$, but these components are not independent. 
\end{enumerate} 
In all cases, $\mathbf Z$ satisfies UVS and WUVS by Theorem \ref{th:closure}. 
In cases (i) and (ii),  $\mathbf Z$ is comonotonic and hence its UVS and WUVS also follow from Proposition \ref{prop:como}.
\end{example}

 Finally, we note another straightforward closure property under dimension reduction:  if $(X_1,\dots,X_n)\in \mathcal L^{\rm WU}_n$,
then 
$(X_1,\dots,X_k)\in \mathcal L^{\rm WU}_k$ for all $k\in [n]$.
The same property does not hold for $\mathcal L^{\rm U}_n$.

\begin{proposition}\label{prop:sub}
 If $(X_1,\dots,X_n)\in \mathcal L^{\rm WU}_n$,
then 
$(X_i)_{i\in K}\in \mathcal L^{\rm WU}_k$ for all $K \subseteq [n]$ with $\vert K \vert =k\le n$.
The same property does not hold for $\mathcal L^{\rm U}_n$ when $n \ge 3$.
\end{proposition}
\begin{proof}
The WUVS case is straightforward by taking $\theta_i=0$ for $i \in [n] \setminus K$.  For UVS, we provide a counterexample.
Let $Y\sim \mathrm{Pareto}(1/2)$ and $U \sim \mathrm{U}(0,1)$ be independent. We take $X_1=(Y+3) \id_{\{0<U \le 1/3\}}$, $X_2=(Y+3) \id_{ \{1/3< U\le 2/3\}}$ and $X_3=(Y+3) \id_{\{2/3<U\le 1\}}$. By construction, we have 
$\VaR_p(X_1+X_2+X_3)=\VaR_p(Y+3)=\VaR_p(Y)+3=(1-p)^{-2}+3$ and $\VaR_{p}(X_i)=\VaR_{3p-2}(Y+3)=(1-p)^{-2}/9+3$ for $2/3<p<1$ and $\VaR_{p}(X_i)=0$ for $0<p\le 2/3$. Hence, it is clear that $\VaR_p(X_1+X_2+X_3)\ge \VaR_p(X_1)+\VaR_{p}(X_2)+\VaR_p(X_3)$ for $0<p \le 2/3$. For  $2/3<p<1$, we have 
\begin{align*}&\VaR_p(X_1+X_2+X_3)-\VaR_p(X_1)-\VaR_{p}(X_2)-\VaR_p(X_3)=\frac{2}{3}(1-p)^{-2}-6>0.
\end{align*}
Therefore, we have $\VaR_p(X_1+X_2+X_3)\ge \VaR_p(X_1)+\VaR_{p}(X_2)+\VaR_p(X_3)$ for all $p\in (0,1)$ and $(X_1, X_2, X_3) \in \mathcal{L}^{\rm U}_3$. On the other hand, for $p \in (1/3,1)$, we have
$\VaR_p(X_1+X_2)=\VaR_{(3p-1)/2} (Y+3)=4(1-p)^{-2}/9+3$.
We can check that, when $p=0.7$, 
$$\VaR_p(X_1+X_2)-\VaR_p(X_1)-\VaR_p(X_2)=\frac{2}{9}(1-p)^{-2}-3<0.$$
Hence $(X_1, X_2) \notin \mathcal{L}^{\rm U}_2$. For the case $n > 3$, we can take $X_i=0$ for $i>3$.
\end{proof}

\subsection{Implications for decision making}

The stochastic dominance relation  in Lemma \ref{lem:st} further leads to the optimal decision in a portfolio optimization problem.  
Let $\rho: \mathcal X^\rho \to \R$ be a risk measure defined on a convex cone $\mathcal X^\rho\subseteq \mathcal L_1^0$ that satisfy the following properties:
\begin{enumerate}
\item[(a)] Monotonicity: $\rho(X)\le \rho(Y)$ if $X\le_{\rm st}Y$;
\item[(a*)] Strict monotonicity: $\rho$ is monotone and $\rho(X)<\rho(Y)$ if $\p(X<Y)=1;$
% $\rho(X)\le \rho(Y)$ if $X\le_{\rm st}Y$ and $\rho(X)= \rho(Y)$ if $X\laweq Y$;
\item[(b)] Comonotonic additivity: $\rho(X+Y)=\rho(X)+\rho(Y)$ if $X$ and $Y$ are comonotonic.
\end{enumerate}
On the space of bounded random variables, \cite{S86} showed that a risk measure that satisfies (a) and (b) follows a Choquet integral form; moreover, distortion risk measures are the only class of law-invariant risk measures that satisfy (a), (b), and $\rho(1)=1$. 
This can be extended to general convex cones (e.g., \citealp{WWW20}).
The class of distortion risk measures includes VaR, ES, and Range-VaR \citep{cont2010robustness} as special cases. On the space of bounded random variables, the two properties imply \citep{S86}
\begin{enumerate}
\item[(c)] Positive homogeneity: $\rho(cX)=c\rho(X)$ for all $c>0$. %\com{I don't think we need this. It is implied by (a)--(b) anyway}
\end{enumerate}
An immediate implication for random vectors satisfying UVS is that any monotone and comonotonic-additive risk measure is also superadditive for such random vectors.
\begin{proposition}\label{prop:sup}
Let $\rho$ be a law-invariant risk measure satisfying (a)--(b) and $(X_1,\dots,X_n) \in (\mathcal{X}^\rho)^n$ satisfy UVS. Then $\rho(\sum_{i=1}^nX_i)\ge \sum_{i=1}^n\rho(X_i)$.
\end{proposition}
\begin{proof}
By Lemma \ref{lem:st}, $X_1^c+\dots+X_n^c\le_{\rm st} X_1+\dots+X_n$ where $(X_1^c,\dots,X_n^c)$ is a comonotonic copy of $(X_1,\dots,X_n)$. Then
$\sum_{i=1}^n\rho(X_i)=\sum_{i=1}^n\rho(X_i^c)=\rho(\sum_{i=1}^nX_i^c)\le \rho(\sum_{i=1}^nX_i)$.
\end{proof}

Suppose that an agent has the following minimization problem 
$$\min_{(\theta_1,\dots,\theta_n)\in\Delta_n}\rho\(\sum_{i=1}^n\theta_iX_i\),$$
where $(X_1,\dots,X_n)$ satisfies WUVS.
For $i\in [n]$, let $\mathbf e_{i,n} $ be the $i$th column vector of the $n\times n$ identity matrix. The proposition below shows that under mild conditions, no diversification is optimal for the agent who faces extremely heavy-tailed risks. 
\begin{proposition}\label{prop:rm}
Let  $\rho$ be a law-invariant risk measure satisfying (a)--(c) and $(X_1,\dots,X_n) \in (\mathcal{X}^\rho)^n$ satisfy WUVS.
    % Suppose that $(X_1,\dots,X_n)$ satisfies WUVS and $\rho$ is a law-invariant risk measure such that it satisfies (a)--(b)   and $\rho(X_i)\in\R$ for all $i\in[n]$.  
      Let $I^*=\argmin_{i\in [n]}\rho(X_i)$. Then 
    $$\mathbf e_{i^*,n}\in \argmin_{(\theta_1,\dots,\theta_n)\in\Delta_n}\rho\(\sum_{i=1}^n\theta_iX_i\)~~~ \mbox{ for each $i^*\in I^*$}.$$ 
   If  further $\rho$ satisfies (a*) and $(X_1,\dots,X_n)$ satisfies strict WUVS, then    $\mathbf e_{i^*,n}$  for $i^*\in  I^*$ are the only optimizers.
\end{proposition}
\begin{proof}
    Let $(X_1^c,\dots,X_n^c)$ be a comonotonic copy of $(X_1,\dots,X_n)$. By Lemma \ref{lem:st}, we have 
    $\theta_1X_1^c+\dots+\theta_n X_n^c\le_{\rm st} \theta_1 X_1+\dots +\theta_n X_n.$ Hence,
    \begin{align*}
        \rho\(\sum_{i=1}^n\theta_iX_i\)\ge  \rho\(\sum_{i=1}^n\theta_iX_i^c\)=\sum_{i=1}^n\rho\(\theta_iX_i^c\)=\sum_{i=1}^n\theta_i\rho\(X_i\)\ge \min_{i\in[n]}\rho(X_i)=\rho(X_{i^*}).
    \end{align*}
    The two equalities are due to comonotonic additivity and positive homogeneity of $\rho$. If $\rho$ is strictly monotone and $(X_1, \dots, X_n)$ satisfies strict WUVS, the first inequality is strict when $(\theta_1,\dots,
    \theta_n)$ has at least two positive components as implied in Remark \ref{rem:strict}. Therefore,    $\mathbf e_{i^*,n}$  for $i^*\in  I^*$ are the only optimizers.
\end{proof}
\begin{remark}
  For most choices of $\mathcal X_\rho$, property (c) automatically follows from (a)--(b). Here we include (c) in Proposition \ref{prop:rm} because our assumption on $\mathcal X_\rho$ is minimal. 
  
\end{remark}

Proposition \ref{prop:rm} suggests that, regardless of the risk measure $\rho$ used, for a vector of risks satisfying WUVS, an optimal portfolio  always concentrates on one asset that has the lowest risk. 
 With strict WUVS, such concentrated portfolios are the only optimal ones.
Therefore, portfolio diversification effect does not take place in the presence of WUVS.

\section{UVS and WUVS for some classes of risks}\label{sec:ind}

This section proceeds to study UVS and WUVS for random vectors with independent  or  {negatively dependent}  components.  In particular, we investigate the marginal distributions of a random vector $(X_1,\dots,X_n)$ that satisfies UVS and WUVS, where $X_1,\dots,X_n$ are independent or  {negatively dependent}. We mainly focus on WUVS, which  is stronger than UVS. 
% Equivalently, by Lemma \ref{lem:st}, we focus on independent risks $X_1,\dots,X_n$ such that the following stochastic dominance relation holds:
% \begin{equation*}\label{eq:SD}
% X_1^c+\dots+X_n^c\le_{\rm st} X_1+\dots+X_n, \tag{\trd{SD}}
% \end{equation*}
% where $(X_1^c,\dots,X_n^c)$ are comonotonic copies of  $(X_1,\dots,X_n)$.

\subsection{Completely subscalable risks}\label{sec:cs}

% In this paper, all random variables are assumed to be non-negative unless stated otherwise. 
We show that $(X_1,\dots,X_n)$ satisfies WUVS if $X_1,\dots,X_n$ are independent completely subscalable risks, a class introduced by  \cite{V25}.
We first define the class of completely subscalable risks.
\begin{definition}%[\cite{V25}]
For a non-negative random variable $X\sim F$, we say $X$ (or $F$) is \emph{completely subscalable}  or belongs to class $\mathcal{CS}$, denoted by $X\in\mathcal{CS}$ (or $F \in \mathcal{CS}$), if $x\overline F(x)$ is increasing on $(0,\infty)$.
\end{definition}
 By Lemma 6.4 of \cite{V25}, if $X \in \mathcal{CS}$,
$\theta \p(X>x) \le \p(\theta X>x)$ for all $\theta \in (0,1)$ and  $x>0$.
\cite{V25} described this as ``scaling down any $\mathcal{CS}$ risk by a
factor $\theta$ reduces the probability of exceeding $x$ less than proportionally." \cite{V25} showed that if an iid random vector $(X_1,\dots,X_n)$ has $\mathcal{CS}$ marginal risks,  $(X_1,\dots,X_n)\in\mathcal{L}^{\rm WU}_n$. Examples of $\mathcal{CS}$ distributions include the Pareto distribution and the Fr\'echet distribution, both with infinite mean.
Some simple and known properties of $\mathcal{CS}$ risks are collected below. 
\begin{proposition}\label{pro:cs}
Let $X\sim F$ and $X\in \mathcal{CS}$. Then
\begin{enumerate}[(i)]
% \item $\theta \overline F(x)\le\overline F(x/\theta)$ for all $x\in(0,\infty)$ and $\theta\in(0,1)$;
\item $\overline F$ is continuous on $(0,\infty)$;
\item $g(X)\in \mathcal{CS}$ where $g$ is an increasing and convex function with $g(0)=0$;
\item if $X\sim F$ has a density function $f$, then the hazard rate $f(x)/{\overline F(x)}\le 1/x$ for almost every $x\in(0,\infty)$.
\end{enumerate}
\end{proposition}
\begin{proof}
Statements (i) and (ii) can be found in Lemmas  6.5 and 6.6 of \cite{V25}. Statement (iii) is trivial.
\end{proof}

Theorem 4.2 of \cite{V25} shows that for independent $\mathcal{CS}$ risks $X_1,\dots,X_n$ and exposure vector $(\theta_1,\dots,\theta_n)\in\Delta_n$, 
\begin{equation}\label{eq:mixture}
X_1\id_{A_1}+\dots+X_n\id_{A_n}\le_{\rm st} \theta_1X_1+\dots+\theta_nX_n,
\end{equation}
where $A_1,\dots,A_n$ are mutually exclusive events, independent of $X_1,\dots,X_n$, and $\p(A_i)=\theta_i$ for all $i\in[n]$. The left-hand side of this inequality is a distribution mixture of $X_1,\dots,X_n$ since its distribution function is $\sum_{i=1}^n\theta_i\p(X_i\le x)$ for $x\ge 0$; a more general version of \eqref{eq:mixture} is shown for the class of super-Fr\'echet risks (see Section \ref{sec:invsub} for the definition) by \cite{CS25} where the generalized mean function is used to mix the distributions. In the case where $X_1,\dots,X_n$ are  iid and $\theta_1=\dots=\theta_n=1/n$, the above inequality becomes $nX_1\le_{\rm st} X_1+\dots+X_n$, and thus $(X_1, \dots, X_n)$ satisfies UVS. We will generalize this observation to non-identically distributed and independent $\mathcal{CS}$ risks. 
 We first show that $(X_1,X_2)$ with independent $\mathcal{CS}$ risks satisfies WUVS.
For simplicity, we say that 
$X$ (or its distribution function $F$) is \emph{doubly continuous} if
$F$ and $F^{-1}$ are both continuous.
 This condition is equivalent to    $F^{-1}$ being
  strictly increasing and continuous; see Appendix A.1 of \cite{MFE15}.

\begin{lemma}\label{lem:strictst}
Let $(X_1, \dots, X_n) \in \mathcal L_n^0$ with doubly continuous components,   and $(X_1^c,\dots,X_n^c)$ be a comonotonic copy of $(X_1,\dots,X_n)$. Then 
 $(X_1, \dots, X_n)$ satisfies strict WUVS if and only if
$$\p(\theta_1X_1^c+\dots+\theta_nX_n^c>t)<\p(\theta_1X_1+\dots+\theta_nX_n>t)$$ for  all $t\in (\sum_{i=1}^n\theta_i \essinf X_i, \sum_{i=1}^n \theta_i  \esssup X_i)$ and all $(\theta_1, \dots, \theta_n) \in \Delta_n$ with at least two positive components.
\end{lemma}
\begin{proof}
 Let $X\sim F$ be doubly continuous and $Y\sim G$. We first show that $\overline F(t)<\overline G(t)$ for all $t\in(\essinf X,\esssup X)$ if and only if $\VaR_p(X)<\VaR_p(Y)$ for all $p\in(0,1)$.  Since $F^{-1}$ is strictly increasing and continuous, $F^{-1}(p)\in(\essinf X,\esssup X)$ and $F(F^{-1}(p))=p$ for any $p\in(0,1)$, and $F^{-1}(F(x))=x$ for any $x\in(\essinf X,\esssup X)$. Note that for any cumulative distribution function $H$, $H^{-1}(u)\le x$ if and only if $u\le H(x)$ where $u\in(0,1)$ and $x$ is in the domain of $H$. This can be equivalently stated as $H^{-1}(u)> x$ if and only if $u> H(x)$.  To see the ``$\Longrightarrow$'' direction, for any $p\in(0,1)$, we have 
$G\(F^{-1}(p)\)<F\(F^{-1}(p)\)=p,$
and thus $G^{-1}(p)>F^{-1}(p)$. To see the ``$\Longleftarrow$'' direction, for any $x\in(\essinf X,\esssup X)$,  we have 
$G^{-1}(F(x))>F^{-1}(F(x))=x,$
and thus $G(x)<F(x)$. 

For all $p\in(0,1)$, we have  
$\VaR_p\(\sum_{i=1}^n\theta_iX_i^c\)=\sum_{i=1}^n\VaR_p\(\theta_iX_i^c\)=\sum_{i=1}^n\VaR_p\(\theta_iX_i\)$. 
It is clear that $\theta_1 X_1^c+\dots+\theta_n X_n^c$ is doubly continuous. We have the desired result.
\end{proof} 
 
\begin{lemma}\label{lem:cs}
 Let  $X_1, X_2\in\mathcal{CS}$ be independent.  Then, $(X_1, X_2)$ satisfies WUVS. {If, furthermore, $X_1$ and $X_2$ are doubly continuous, then  $(X_1, X_2)$ satisfies strict WUVS.}
\end{lemma}
\begin{proof}
    Let  $X_1\sim F_1$ and $X_2\sim F_2$ be independent and $(X_1^c,X_2^c)$ be a comonotonic copy of $(X_1,X_2)$. Then $X_1+X_{2}$ has the same distribution as $F_1^{-1}(U_1)+F_{2}^{-1}(U_2)$ and $X_1^c+X_2^c$ has the same distribution as $F_1^{-1}(U)+F_2^{-1}(U)$, where $U_1,U_2,$ and $U$ are all uniform random variables on $(0,1)$ with $U_1$ and $U_2$ being independent.
     Note that for any cumulative distribution function $F$, $F^{-1}(u)\le x$ if and only if $u\le F(x)$ where $u\in(0,1)$ and $x$ is in the domain of $F$. Since stochastic order is preserved under increasing transforms,  it suffices to show 
 $$G(F_1^{-1}(U_1)+F_2^{-1}(U_2))\ge_{\rm st}U,$$
 where $G$ has left quantile function $G^{-1}=F_1^{-1}+F_2^{-1}$. Fix $x\in(0,1)$.
 Write $y=G^{-1}(x)$ and note that $F_1^{-1}(U_1)+F_2^{-1}(U_2)\ge y$ implies at least one of $U_1$ and $U_2$ is at least $x$. Assume $F_1^{-1}(x)>0$ and $F_2^{-1}(x)>0$.
 By independence of $U_1$ and $U_2$,   we have
\begin{align}\label{eq:csproof}
&\p\(G(F_1^{-1}(U_1)+F_2^{-1}(U_2))\ge x\) \nonumber\\
&=\p\(F_1^{-1}(U_1)+F_2^{-1}(U_2)\ge y \) \nonumber\\
& = \p(U_1, U_2>x) +\p(  U_1 \le x, F_2^{-1}(U_2)\ge y-F_1^{-1}(U_1)) +\p(  U_2 \le x, F_1^{-1}(U_1)\ge  y -F_2^{-1}(U_2)) \nonumber\\
& \ge \p(U_1 , U_2>x) +\p(  U_1 \le x, F_2^{-1}(U_2)>y)+\p(  U_2 \le x, F_1^{-1}(U_1)>y)\\
&=(1-x)^2+x \(\overline F_1(y)+\overline F_2(y)\) \nonumber\\ 
&\ge (1-x)^2+x \(\frac{F_1^{-1}(x)}{y}\overline F_1(F_1^{-1}(x))+\frac{F_2^{-1}(x)}{y}\overline F_2(F_2^{-1}(x))\) \nonumber\\
&=1-x.\nonumber
\end{align}
 The last inequality is because $F_1,F_2\in \mathcal{CS}$ and the last equality is due to the continuity of $F_1,F_2$ (Proposition \ref{pro:cs} (i)). If $F_i^{-1}(x)=0$, $\overline F_i(y)\ge (F_i^{-1}(x)/y)(1-x)$ trivially holds and so does the above inequality chain. Hence,  $(X_1, X_2)$ satisfies  UVS. Note that if $X_1, X_2\in \mathcal{CS}$, then $\theta_1 X_1, \theta_2 X_2 \in \mathcal{CS}$ for all $(\theta_1, \theta_2) \in \Delta_2$. Hence, $(X_1, X_2)$ satisfies WUVS. The strictness statement holds as inequality \eqref{eq:csproof} is strict provided that each $X_i$ is doubly continuous and $U_1$ and $U_2$ are independent. By  Lemma \ref{lem:strictst},  we obtain the desired result. 
\end{proof}
To extend Lemma \ref{lem:cs} to $n$-dimensional random vectors, we need the following result.
\begin{proposition}\label{prop:cs}
Let  $X_1, \dots, X_n\in\mathcal{CS}$ be comonotonic. Then, $X_1+\dots+X_n\in\mathcal{CS}$. 
\end{proposition}
\begin{proof}
Let $H$ be the cumulative distribution function of $X_1+\dots+X_n$ and $X_i\sim F_i$ for all $i\in[n]$. Then $H^{-1}=F_1^{-1}+\dots+F_{n}^{-1}$ and $H^{-1+}=F_1^{-1+}+\dots+F_{n}^{-1+}$. 
Next, we show that $u\overline H(u)$ is increasing on $u\in(0,\infty)$. Let $0<u_1<u_2$. If $H(u_1)=H(u_2)$, clearly $u_1\overline H(u_1)\le u_2\overline H(u_2)$. We focus on the case where $H(u_1)<H(u_2)$. Let $t_1=H(u_1)$ and $t_2=H(u_2)$ where $t_1<t_2$. To show $u_1\overline H(u_1)\le u_2\overline H(u_2)$, since $u_1\overline H(u_1)\le H^{-1+}(t_1)(1-t_1)$ and $H^{-1}(t_2)(1-t_2)\le u_2\overline H(u_2)$, it suffices to show 
$H^{-1+}(t_1)(1-t_1)\le H^{-1}(t_2)(1-t_2),$
which can be equivalently written as
$$F_1^{-1+}(t_1)(1-t_1)+\dots+F_{n}^{-1+}(t_1)(1-t_1)\le F_1^{-1}(t_2)(1-t_2)+\dots+F_{n}^{-1}(t_2)(1-t_2).$$
Since $X_1,\dots,X_n\in\mathcal{CS}$,  $s\overline F_i(s)$, $i\in[n]$, are increasing on $s\in(0,\infty)$. Moreover, as $F_i$, $i\in[n]$, are continuous, we have $$F_i^{-1+}(t_1)(1-t_1)=F_i^{-1+}(t_1)\overline F_i(F_i^{-1+}(t_1))\le F_i^{-1}(t_2)\overline F_i(F_i^{-1}(t_2))=F_i^{-1}(t_2)(1-t_2).$$
This completes the proof.
\end{proof}

% We recall that stochastic order is preserved under convolution from Theorem 1.A.3 (b) of \cite{SS07} in the following lemma.
% \begin{lemma}\label{lem:con}
% Let $(X_1, \dots, X_n)$ and $(Y_1, \dots, Y_n)$ be two random vectors with independent components. If $X_i \lst Y_i$ for all $i \in [n]$, then 
% $\sum_{i=1}^n X_i\lst \sum_{i=1}^n Y_i$.
 
Now we can state  the main result of this section.

\begin{theorem}\label{thm:cs}
 Let  $X_1, \dots, X_n\in\mathcal{CS}$ be independent.  Then, $(X_1, \dots, X_n)$ satisfies WUVS. If, furthermore, $X_1,\dots,X_n$  are doubly continuous, then  $(X_1, \dots,X_n)$ satisfies strict WUVS.
\end{theorem}
\begin{proof}
Let $(Y_1,\dots,Y_n)$ be a comonotonic copy of $(X_1,\dots,X_n)$ and be independent of $(X_1,\dots,X_n)$. The case  $n=2$ is shown by Lemma \ref{lem:cs}. For the case $n>2$, 
assume 
$(X_1, \dots, X_{n-1}) \in \mathcal{L}^{\rm U}_{n-1}$. %Let $X_n$ be independent of $X_1,\dots,X_{n-1}$ and $Y_1,\dots,Y_{n-1}$. 
%As stochastic order is preserved under convolution (e.g.,  Theorem 1.A.3 (b) of \cite{SS07}), 
 We have 
\begin{equation}\label{eq:cs1}Y_1+\dots+Y_{n-1}+Y_{n} \le_{\rm st} Y_1+\dots+Y_{n-1}+X_{n}\le_{\rm st} X_1+\dots+X_{n-1}+X_n,\end{equation}
where the first inequality is implied by Lemma \ref{lem:cs} and Proposition \ref{prop:cs}, and the second inequality is due to the fact that  stochastic order is preserved under convolution  \citep[Theorem 1.A.3 (b)]{SS07}. 
% By Lemmas \ref{lem:con} and  \ref{lem:1}, we have
% $$Y_1+\dots+Y_{n-1}+Y_{n} \le_{\rm st} Y_1+\dots+Y_{n-1}+X_{n}\le_{\rm st} X_1+\dots+X_{n-1}+X_n.$$
By Lemma \ref{lem:st}, $(X_1, \dots, X_n) \in \mathcal{L}^{\rm U}_n$. Then, $(X_1, \dots, X_n)$ satisfies  UVS  by induction. Let $(\theta_1, \dots, \theta_n) \in \Delta_n$. By Proposition \ref{pro:cs},  $\theta_1X_1, \dots, \theta_nX_n \in \mathcal{CS}$. Hence,  $(\theta_1X_1, \dots, \theta_n X_n)$  satisfies  UVS. For the strictness statement, it suffices to assume $\theta_1,\dots,\theta_n>0$.  If $X_1,\dots,X_n$ are doubly continuous, so are $\theta_1Y_1+\dots+\theta_{n-1}Y_{n-1}$ and $\theta_nY_n$. By Lemmas \ref{lem:strictst} and \ref{lem:cs}, $\p(\sum_{i=1}^n\theta_iY_i>t) < \p(\sum_{i=1}^{n-1}\theta_iY_i+\theta_nX_{n}>t)\le\p(\sum_{i=1}^n\theta_iX_i>t)$ for all $t\in (\sum_{i=1}^n \theta_i\essinf X_i, \sum_{i=1}^n\theta_i\esssup X_i)$.  Lemma \ref{lem:strictst} completes the proof. 
\end{proof}
 Together with inequality \eqref{eq:mixture} and Theorem \ref{thm:cs}, we have the following corollary.
 \begin{corollary}
     Let  $X_1,\dots,X_n\in\mathcal{CS}$ be independent with $X_i\sim F_i$ for $i\in[n]$,  $U$ be uniformly distributed on $(0,1)$,  and $(\theta_1,\dots,\theta_n)\in\Delta_n$. Then 
   $$\max\(H^{-1}(U),G^{-1}(U) \) \le_{\rm st} \theta_1X_1+\dots+\theta_nX_n, $$
   where $H(x)=\sum_{i=1}^n\theta_iF_i(x)$ for $x\ge 0$ and $G^{-1}(p)=\sum_{i=1}^n\theta_iF_i^{-1}(p)$ for $p\in(0,1)$.
 \end{corollary}
 \begin{proof}
     By Lemma \ref{lem:st}, Theorem \ref{thm:cs}, and \eqref{eq:mixture}, for all $p\in(0,1)$
     $$\max\(H^{-1}(p),G^{-1}(p) \)\le \VaR_p\(\sum_{i=1}^n\theta_iX_i\).$$
     Plugging $U$ into the quantile functions, we have the desired result.
 \end{proof}
 
% Theorem \ref{thm:cs} and Proposition \ref{pro:cs} together imply the following corollary.

% \begin{corollary}\label{cor:cs}
%   Let  $X_1, \dots, X_n\in\mathcal{CS}$ be independent.  Then, $(X_1, \dots, X_n)$ is WUVS.
% \end{corollary}
%\begin{proof}
% Let $(\theta_1, \dots, \theta_n) \in \Delta_n$. By Proposition \ref{pro:cs},  $\theta_1X_1, \dots, \theta_nX_n \in \mathcal{CS}$. As $\theta_nX_1, \dots, \theta_nX_n$ are independent, by Theorem \ref{thm:cs}, we have $(\theta_1X_1, \dots, \theta_n X_n) \in \mathcal{L}^{\rm US}_n$, which completes the proof.
% \end{proof}
\begin{remark}[Majorization inequality]
Let risks $X_1,\dots,X_n\in\mathcal{CS}$ be iid and $w_1,\dots,w_n\ge 0$. Theorem \ref{thm:cs} implies 
\begin{equation}\label{eq:iid}
\(\sum_{i=1}^nw_i\)X_1\le_{\rm st}w_1X_1+\dots+w_nX_n.\end{equation}
 For two vectors $ \(\gamma_1,\dots,\gamma_n\)$ and $  \(\eta_1,\dots,\eta_n\)$ in $\R^n$,  $ \(\gamma_1,\dots,\gamma_n\)$ is dominated by $  \(\eta_1,\dots,\eta_n\)$ in majorization order if 
$$
\sum^n_{i=1}\gamma_i =\sum^n_{i=1}\eta_i \mbox{~~~and~~~}
\sum^k_{i=1} \gamma_{(i)} \ge \sum^k_{i=1} \eta_{(i)}\ \ \mbox{for}\ k\in [n-1],
$$
where $\gamma_{(i)}$ and $\eta_{(i)}$ represent the $i$th smallest order statistics of $  \(\gamma_1,\dots,\gamma_n\)$ and $  \(\eta_1,\dots,\eta_n\)$, respectively. Let $\theta_1,\dots,\theta_n\ge0$ such that $(\theta_1,\dots,\theta_n)$ is larger than $(w_1,\dots,w_n)$ in majorization order. A stronger result than inequality \eqref{eq:iid} is 

\begin{equation}\label{eq:maj}
\theta_1X_1+\dots+\theta_nX_n\le_{\rm st}w_1X_1+\dots+w_nX_n.
\end{equation}
\cite{CHSZ25} showed inequality \eqref{eq:maj} for $X_1,\dots,X_n\sim F$ where $F$ is the inverted concave distribution, i.e.,  $\overline F(1/x)$ is concave on $(0,\infty)$. It is easy to see that the class of inverted concave distributions is a subset of $\mathcal{CS}$. A natural question is whether \eqref{eq:maj} holds for iid
 $X_1,\dots,X_n\in \mathcal{CS}$. %The ``$\Longleftarrow$'' direction is clearly true but the ``$\Longrightarrow$'' direction does not hold in general, as illustrated by the following example.
The following example shows that this is not true.
\begin{example}\label{ex:maj} Let risks $X,X_1,\dots,X_n\sim F$ be iid where
$$
F(x)=
\begin{cases}
0,~~\mbox{for $x<1$},\\
1-(1/x),~~\mbox{for $1\le x<2$},\\
1/2,~~\mbox{for $2\le x<3$},\\
1-3/(2x),~~\mbox{for $x\ge 3$}.
\end{cases}
$$
It is easy to check $x\overline F(x)$ is increasing on $(0,\infty)$ and thus $X\in\mathcal{CS}$. By Theorem \ref{thm:cs}, inequality \eqref{eq:iid} holds. However, Example 4.5 of \cite{ZZSH25} shows 
$$ \frac{1}{4}X_1+\frac{3}{4}X_2 \nleq_{\rm st}  \frac{2}{5}X_1+\frac{3}{5}X_2.$$
As $(2/5,3/5)$ is smaller than $(1/4,3/4)$ in majorization order, this shows that inequality \eqref{eq:maj} does not hold for $X_1,\dots,X_n$.
\end{example}
\end{remark}

\subsection{Super-$F$ classes }\label{sec:invsub}

As demonstrated by Example \ref{ex:pareto}, Theorems \ref{th:convex} and  \ref{pro:SAw} are useful in generalizing existing results for iid random vectors in $\mathcal{L}^{\rm WU}_n$ to random vectors with independent but not necessarily identically distributed components. To make the most of Theorems \ref{th:convex} and  \ref{pro:SAw}, it suffices to consider the largest class of marginal distributions for iid random vectors in  $\mathcal{L}^{\rm WU}_n$. However, an explicit characterization of this largest class is not known. To the best of our knowledge, the two maximal classes admitting explicit expressions of this kind are the super-Cauchy class of \cite{Muller25} and the inverted subadditive class of \cite{ALO25}. The two classes are not subsets of one another. We start with the super-Cauchy class, which is defined via strictly increasing and convex transformations of standard Cauchy risks.
\begin{definition}
A random variable $X$ is \emph{super-$F$} (or has a super-$F$ distribution) if $X\laweq f(Y)$ for some strictly increasing and convex function $f$ and $Y \sim F$.  If further $f$ is strictly convex, then we say $X$ is \emph{strictly super-$F$}. 
\end{definition}

% \begin{example}
%  The cumulative distribution function of the Fr\'echet distribution function with tail parameter $\alpha\in(0,\infty)$, denoted by Fr\'echet$(\alpha)$, is given by 
%  $$F(x)=\exp\(-\frac{1}{x^\alpha}\),~~~ x \ge 0.$$  If $\alpha\le 1$, the Fr\'echet distribution is completely
% subscalable and has infinite mean.
% \end{example}

 The set of strictly super-$F$ risks is a subclass of super-$F$ risks. Special cases of super-$F$ risks include:
\begin{enumerate}[(i)]
\item If $F={\rm Pareto}(1)$, then $X$ is \emph{super-Pareto}, denoted by $X\in\mathcal{S}_P$; if in addition, $\essinf X= 0$, we write $X\in\mathcal{S}_P^+$;
\item If $F$ is the cumulative distribution function of the Fr\'echet distribution with tail parameter 1, given by
 $$F(x)=\exp\(-\frac{1}{x}\),~~~ x \ge 0.$$
then $X$ is \emph{super-Fr\'echet}, denoted by $X\in\mathcal{S}_F$;  if in addition, $\essinf X= 0$, we write $X\in\mathcal{S}_F^+$;
\item If $F$ is the cumulative distribution function of the standard Cauchy distribution, given by
$$F(x)=\frac{1}{\pi}\arctan(x)+\frac{1}{2},~~~~~x\in\R,$$
then $X$ is \emph{super-Cauchy}, denoted by $X\in\mathcal{S}_C$.
\end{enumerate}
The classes  of risks $\mathcal{S}_P$, $\mathcal{S}_F$, and $\mathcal{S}_C$ were proposed by \cite{CEW25}, \cite{CS25}, and \cite{Muller25}, respectively. In \cite{CEW25} and \cite{CS25}, random vectors in $\mathcal{L}^{\rm WU}_n$ with negatively dependent components are also considered.  It is known that $\mathcal{S}_P^+\subseteq \mathcal{S}_F^+$ and $\mathcal{S}_P\subseteq \mathcal{S}_F \subseteq \mathcal{S}_C$.  The authors showed that $(X_1,\dots, X_n)$ satisfies WUVS if $X_1, \dots, X_n$ are iid risks in $\mathcal{S}_P$ \citep{CEW25}, $\mathcal{S}_F$ \citep{CS25}, or $\mathcal{S}_C$ \citep{Muller25}. By Theorems \ref{th:convex} and  \ref{pro:SAw}, we generalize the result of \cite{Muller25} to independent but not necessarily identically distributed super-Cauchy and   strictly super-Cauchy risks in the following result.  
\begin{proposition}\label{cor:SC}

If $X_1, \dots, X_n$ are  independent (strictly) super-Cauchy risks, then $(X_1, \dots, X_n)$ satisfies (strict) WUVS. 
\end{proposition}
\begin{proof}
 Since $X_1, \dots, X_n$ are independent and super-Cauchy, there exist  iid standard Cauchy risks $ Y_1,\dots, Y_n$ such that $X_i=f_i(Y_i)$, $i\in [n]$, for some strictly increasing and convex functions $f_1, \dots, f_n$. For iid standard Cauchy risks $Y, Y_1,\dots, Y_n$, it is well-known that $\theta_1Y_1+\dots+\theta_n Y_n\laweq Y$ for any $(\theta_1, \dots, \theta_n) \in \Delta_n$ (see Property 1.2.1 of \citealp{samoradnitsky2017stable}) and hence 
 $(X_1, \dots, X_n)\in \mathcal{L}^{\rm WU}_n$ by Theorem \ref{th:convex} (i). If $X_1, \dots, X_n$ are further strictly super-Cauchy, $f_1, \dots, f_n$ are  strictly convex. Moreover, $(Y_1,\dots,Y_n)$ has a positive density on $\R^n$ by independence. Applying Theorem \ref{pro:SAw} (i) completes the proof. 
\end{proof}

% We consider a class of super-$\mathcal{IS}$ risks such that any risk in this class can be obtained via applying some increasing and convex transforms to some $\mathcal{IS}$ risk. Note that the class $\mathcal{IS}$ is closed under strictly increasing and convex transformation, as implied by Theorem 2.8 of \cite{ALO25}. The following lemma will be useful to show the main result.

% Some well-known heavy-tail example of super-$F$ distribution including super-Pareto ($F$ is Pareto$(1)$), super-Cauchy ($F$ is Cauchy) and super-Fr{\'e}chet ($F$ is Fr{\'e}chet). 

% Therefore, if  an iid $(X_1, \dots, X_n) \in \mathcal{L}^{\rm US}_n^{\mathrm w}$ with identical marginal $F$, then independent $(X_1, \dots, X_n) \in \mathcal{L}^{\rm US}_n^{\mathrm w}$ where $X_i$ is super-$F$ for all $i \in [n]$. 
%\begin{example}
Next, we apply Theorems \ref{th:convex} and  \ref{pro:SAw} to the inverted subadditive class of \cite{ALO25}, defined below. 
\begin{definition}%[\cite{ALO25}]
For a random variable $X\sim F$, we say $X$ (or $F$)  is \emph{inverted
subadditive} or belongs to class $\mathcal{IS}$, denoted by $X\in\mathcal{IS}$ (or $F \in \mathcal{IS}$), if $\overline F(1/x)$ is subadditive on $(0,\infty)$.
\end{definition}
Both $\mathcal{S}_C$ and $\mathcal{IS}$ include $\mathcal{S}_P^+$ and $ \mathcal{S}_F^+$ as subsets; see the next section for more details on the relations of these classes. 
Note that the class $\mathcal{IS}$ is closed under strictly increasing and convex transformations, as implied by Theorem 2.8 of \cite{ALO25}. 
 \cite{ALO25} showed that $(X_1, \dots, X_n)$ with iid components from $\mathcal{IS}$ satisfies WUVS. Theorems \ref{th:convex} and  \ref{pro:SAw} further extend this result to independent but not necessarily identically distributed risks in a subclass of $\mathcal{IS}$ where the risks are strictly increasing and convex transformations of the same $\mathcal{IS}$ risk.

\begin{proposition}\label{cor:IS}

      Let  $X_1, \dots, X_n$ be independent and super-$F$  for some $F\in\mathcal{IS}$.  Then $(X_1, \dots, X_n)$ satisfies WUVS.  Moreover, if $X_1, \dots, X_n$ are further strictly super-$F$ and $Y\sim F$ is atomless with positive density on $(\essinf Y,  \infty)$, then $(X_1, \dots, X_n)$ satisfies strict WUVS.
\end{proposition}
 
\begin{proof}
If $X_1, \dots, X_n$ are independent and super-$F$, there exist iid $Y_1, \dots, Y_n \in \mathcal{IS}$ with $Y_i \sim F$  and strictly increasing and convex functions $f_1, \dots, f_n$ such that $X_i=f_i(Y_i)$.
 \citet[Theorem 3.1]{ALO25} showed that iid $(Y_1, \dots, Y_n)\in \mathcal{L}^{\rm WU}_n$ if $Y_1 \in \mathcal{IS}$. By Theorem \ref{th:convex} (i), we have $(X_1, \dots, X_n)\in \mathcal{L}^{\rm WU}_n$. If further $X_1, \dots, X_n$ are 
 strictly super-$F$, then $f_1, \dots,f_n$ are further strictly  convex functions. Since  $Y_i$ is atomless and has positive density on $(\essinf Y,  \infty)$, $(Y_1, \dots, Y_n)$ is atomless and  has positive density on $  (\essinf Y, \infty)^n$ by independence.
Together with   Theorem \ref{pro:SAw} (i),  we have the desired result. 
\end{proof}

 %For a given $F \in \mathcal{IS}$, the super-$F$ class is a proper subset of $\mathcal{IS}$. That is, not all risks in $\mathcal{IS}$ are strictly increasing and convex transformations of the same risk in $\mathcal{IS}$.

 % independent $(X_1, \dots, X_n) \in \mathcal{L}^{\rm US}_n^{\mathrm w}$ if $X_i \in S_C$ for all $i \in [n]$.
%\end{example}
% \begin{corollary}\label{cor:SC}
%   Let  $X_1, \dots, X_n$ be independent and super-Cauchy.  Then $(X_1, \dots, X_n)\in \mathcal{SA}^{\mathrm w}_n$.
% \end{corollary}
%

We further apply Theorems \ref{th:convex} and \ref{pro:SAw} to random vectors in $\mathcal{L}^{\rm WU}_n$ with negatively dependent components. The following class of distributions is introduced by \cite{CS25}, containing the super-Fr\'echet class.

\begin{definition}
For a random variable $X\sim F$, we say $X$ belongs to class $\mathcal{D}$ if $-\log F(1/x)$ is subadditive on $(0,\infty)$.
\end{definition}

A random vector $(X_1,\dots,X_n)$ is \emph{negatively lower orthant dependent} (NLOD)  if for all $x_1,\dots,x_n \in\R$,  $\p(X_1\le x_1,\dots,X_n\le x_n)\le \prod_{i=1}^n\p(X_i\le x_i)$. For instance, normal random vectors with non-positive correlation coefficients are NLOD;
 see \cite{block1982some} for more details on this notion. We have the following result. 
 \begin{proposition}\label{prop:D}
 Let  $X_1, \dots, X_n$ be NLOD and super-$F$  for some $F\in\mathcal{D}$. Then $(X_1, \dots, X_n)$ satisfies WUVS. Moreover, if each $X_i$, $i\in[n],$ is strictly super-$F$ and $(X_1, \dots, X_n)$ is atomless and has positive density on $\prod_{i=1}^n(\essinf X_i,\infty)$,  then $(X_1, \dots, X_n)$ satisfies strict WUVS.
 \end{proposition}
 \begin{proof}
 Theorem 1 of \cite{CS25} shows that random vectors with NLOD and identically distributed $\mathcal D$ components satisfy WUVS. Note that the NLOD property is preserved under strictly increasing transforms applied to the components of a random vector by \citet[Theorem 3.2]{block1982some}. If $(X_1, \dots, X_n)$ is atomless and has positive density on $\prod_{i=1}^n(\essinf X_i,  \infty)$, then $(Y_1, \dots, Y_n)=(f_1^{-1}(X_1), \dots, f_n^{-1}(X_n))$ is also atomless and has positive density on $\prod_{i=1}^n (\essinf Y_i,  \infty)$ for strictly increasing and strictly convex functions $f_1, \dots, f_n$.
 Following a similar argument to the proof of Proposition \ref{cor:IS}, Theorem \ref{th:convex} (i) and Theorem \ref{pro:SAw} (i) complete the proof.
 \end{proof} 

  We note that  the  statements in Propositions \ref{cor:IS} and \ref{prop:D}
 cannot be stated for all $X_1,\dots,X_n$ with distributions in $\mathcal{D}$ and $\mathcal{IS}$. 
Indeed, in Example \ref{ex:counter-D} of the next section, we construct a distribution $F_{L,c} \in \mathcal{D} \subseteq \mathcal{IS}$ such that for $X_1\sim F_{1,1/500}$ and $X_2\sim F_{10,1/5}$,   $(X_1, X_2) \notin \mathcal L_2^{\rm U}$; hence  $X_1$ and $X_2$ do not belong to the same super-$F$ class, although their distributions are both in $\mathcal{D}$. 
 
\section{Relations of some classes of risks}\label{sec:relation}  
Random vectors $(X_1,\dots,X_n)\in \mathcal{L}^{\rm WU}_n$
 have been studied under various assumptions of the marginal distributions and dependence structure. In addition to the classes of risks mentioned before (i.e., $\mathcal{CS}$, $\mathcal{IS}$, $\mathcal{S}_P$, $\mathcal{S}_F$, $\mathcal{S}_C$, and $\mathcal D$), we list below another class of risks in $\mathcal{L}^{\rm WU}_n$.

\begin{definition}
For a random variable $X\sim F$, we say
%\begin{enumerate}[(i)]
% \item
%  $X$ belongs to class $\mathcal{S}_P$ if $1/\overline F(x)$ is concave in $(0,\infty)$;
 % \item
 % $X$ belongs to class $\mathcal{S}_F$ if $-1/\log F(x)$ is concave in $(0,\infty)$;
%\item
 $X$ belongs to class $\mathcal{M}$ if $x\log F(x)$ is decreasing on $(0,\infty)$.
% \item
% $X$ belongs to class $\mathcal{IS}$ if $\log F(1/x)$ is super-additive in $(0,\infty)$;
 % \item
 % $X$ belongs to class $\mathcal{IS}$ if $\overline F({1}/x)$ is sub-additive in $(0,\infty)$;
 % \item
 %  $X$ belongs to class $\mathcal{S}_C$ if $\tan(\pi F(x)-\pi/2)$ is concave in $(\essinf F,\infty)$;
%
  %\item $X$ belongs to class $\mathcal{IS}$ if $\bar F(1/x)$ is sub-additive in $(0,\infty)$.
%\end{enumerate}
\end{definition}

% The cumulative distribution function for ${\rm Pareto}(\alpha)$  with parameters $\alpha>0$ is  given by 
% $$P_{\alpha}(x)=1-\left(\frac{1}{x}\right)^\alpha, ~~~ x \ge 1.$$
% For random variables $X$ and $Y$, we write $X \le_{\rm st} Y$ if $\p(X>x)\le \p(Y<x)$ for all $x \in \R$.
% \cite{ALO25} shows that for iid  inverted subadditive risks $(X_1, \dots, X_n) \in \mathcal{L}^{\rm US}_n$. We will generalize this observation to a subclass of inverted subadditive risks that are independent and possibly non-identically distributed.
%  
 Class $ \mathcal{M}$ was introduced by \cite{M25}.  Table \ref{tab:example} summarizes some assumptions that have been considered for $(X_1,\dots,X_n)\in \mathcal{L}^{\rm WU}_n$ from the existing literature. Notably, only \cite{M25} considered non-identically distributed risks, and only $\mathcal{S}_C$ includes risks that are not bounded from below. Theorem \ref{thm:cs} and Propositions  \ref{cor:SC} and \ref{cor:IS} have extended the existing results to non-identically distributed $\mathcal{CS}$ risks, super-Cauchy risks, and some $\mathcal{IS}$ risks. All notions of negative dependence in Table \ref{tab:example} include independence as a special case; we refer to the corresponding papers for the definitions and relations of these notions of dependence.
\begin{table}[h!]
    \centering
    \caption{Results for $(X_1,\dots,X_n)\in \mathcal{L}^{\rm WU}_n$ in existing literature 
and in this paper; $\mathcal{L}^{\rm IM}$ means that the result requires identical marginal distributions.}
\medskip
    \label{tab:example} 
\renewcommand{\arraystretch}{1.4}
    \begin{tabular}{l|l|l|l} 
        Marginal class & Paper or result & Dependence structure & $\mathcal{L}^{\rm IM}$  \\ 
        \hline
                 $\mathcal{CS}$ &  \cite{V25} & independence & yes  \\ 
                 \hline
                 $\mathcal{IS}$ & \cite{ALO25} & independence & yes  \\ 
                 \hline
        $\mathcal{S}_P$ &  \cite{CEW25} & weak negative association & yes  \\ 
        \hline
        $\mathcal{S}_F$ &  \cite{CS25} & negative lower orthant dependence & yes \\ 
        \hline
                 $\mathcal{S}_C$ &  \cite{Muller25} & independence & yes  \\ 
                 \hline
$\mathcal{D}$ &  \cite{CS25} & negative lower orthant dependence  & yes  \\ 
         \hline
        $\mathcal{M}$ &  \cite{M25} &  negative simplex dependence  & no  \\ 
        \hline
      $\mathcal{S}_C$ &  Proposition \ref{cor:SC} & independence & no  \\ 
                 
 \hline
    $\mathcal{CS}$&   Theorem \ref{thm:cs} & independence & no \\ 
                   \hline
     super-$F$,       $F\in \mathcal{IS}$ &  Proposition \ref{cor:IS} & independence & no  \\ 
                        \hline
     super-$F$,       $F\in \mathcal{D}$ &  Proposition \ref{prop:D} & negative lower orthant dependence  & no  \\ 
     
                 \hline
    \end{tabular}
\end{table}

The next result presents the inclusion relations among the aforementioned classes of risks and shows that Theorem \ref{thm:cs} generalizes the distributions assumed by \cite{M25} for $(X_1,\dots,X_n)\in \mathcal{L}^{\rm WU}_n$ where $X_1,\dots,X_n$ are independent.  For two sets $A$ and $B$, we say $A$ and $B$ are incomparable if $A\not\subseteq B$ and $B\not\subseteq A$.

\begin{proposition}
We have 
\begin{enumerate}[(i)]
\item
$\mathcal{S}_P^+\subsetneq \mathcal{S}_F^+\subsetneq\mathcal{M}\subsetneq \mathcal{CS} \subsetneq \mathcal{IS}$;
\item
$\mathcal{S}_P^+\subsetneq \mathcal{S}_F^+\subsetneq\mathcal{M}\subsetneq \mathcal{D} \subsetneq \mathcal{IS}$; $\mathcal {CS}$ and $\mathcal{D}$ are incomparable;
\item
$\mathcal{S}_P\subsetneq \mathcal{S}_F\subsetneq\mathcal{S}_C$;
\item $\mathcal{S}_C$ and $\mathcal{CS}$ (resp.~$\mathcal{M}$, $\mathcal{IS}$, and $\mathcal{D}$) are incomparable.
% \item $\mathcal {D}_4\not\subseteq \mathcal{S}_C$ and $\mathcal{S}_C\not\subseteq\mathcal {D}_4$; $\mathcal {CS}\not\subseteq \mathcal{S}_C$ and $\mathcal{S}_C\not\subseteq\mathcal {CS}$; $\mathcal {D}_3\not\subseteq \mathcal{S}_C$ and $\mathcal{S}_C\not\subseteq\mathcal {D}_3$.
\end{enumerate}
\end{proposition}
\begin{proof}
Relations $\mathcal{S}_P^+\subsetneq \mathcal{S}_F^+ $ and $\mathcal{S}_P\subsetneq \mathcal{S}_F $ can be seen by Example 4 of \cite{CS25} and $ \mathcal{S}_F \subsetneq \mathcal{S}_C$ was shown by Theorem 2.10 of \cite{Muller25}.
% Relation $\mathcal{S}_P\subseteq \mathcal{S}_F \subseteq \mathcal{S}_C$ has been shown by Theorem 2.10 of \cite{Muller25} and $\mathcal{S}_P^+\subseteq \mathcal{S}_F^+ $ has been shown in Example 4 of \cite{CS25}. 
Thus, we have (iii). 
To see $\mathcal{S}_F^+\subseteq\mathcal{M}$, note that the Fr\'echet distribution with tail parameter 1 is in $\mathcal{M}$ (Example 3.11 of \cite{M25}) and that for any $X\in\mathcal{M}$, if $g$ is strictly increasing and convex with $g(0)=0$, then $g(X)\in\mathcal{M}$ (see the proof of Proposition 3.16 of \cite{M25}).  To see $\mathcal{S}_F^+\subsetneq\mathcal{M}$, note that $\mathcal{M}$ contains some log-Cauchy distributions that are not in $\mathcal{S}_F^+$. Let the log-Cauchy distribution have the following cumulative distribution function 
$$H(x)=\frac 12 +\frac 1 \pi\arctan(\log(x)),\quad x>0.$$ By Example 3.11 (6) of \cite{M25}, $H\in \mathcal M$. Note that $H\in\mathcal{S}_F^+$ if and only if $g(x)=-1/\log H(x)$ is strictly increasing and concave on $(0,\infty)$. We have 
$$g''(x)
=
\frac{
2+\log H(x)
\left(
1+\pi H(x)\left(1+(\log x)^2+2\log x\right)
\right)
}{
\pi^2 x^2
\left(1+(\log x)^2\right)^2
H(x)^2
(-\log H(x))^3
}$$
and $g''(1)>0$. Hence, $H$ is not super-Fr\'echet.
%Relation $\mathcal{IS}\subseteq\mathcal{S}_C$  is in Theorem 4.13 of \cite{ALO25}.

For (i), we show $\mathcal{M}\subseteq \mathcal{CS} \subsetneq \mathcal{IS}$.
Let $X\sim F$. Suppose that $X\in \mathcal{M}$. Then for $0< x\le y$, $y\log F(y)\le x\log F(x)$, which implies $F(x)\ge F(y)^{y/x}$. Therefore,
\begin{align*}
\overline F(x)\le 1-F(y)^{y/x}\le \frac{y}{x}(1-F(y)), 
\end{align*}
where the last inequality is because $1-t^r\le r(1-t)$ holds for all $t\in (0,1)$ and $r\ge 1$.
Thus $x\overline F(x)\le y\overline F(y)$ and $\mathcal{M}\subseteq \mathcal{CS}$. Next, assume $X\in \mathcal{CS}$. Then $s\overline F(s)$ is increasing for all $s\ge 0$, or equivalently $\overline F(1/s)/s$ is decreasing on $(0,\infty)$, i.e., $\overline F(1/\cdot)$ is anti-star-shaped; a function $f$ on $[0,\infty)$ is said to be anti-star-shaped if $f(0)=0$ and $f(x)/x$ is decreasing. It is well-known that any anti-star-shaped function is subadditive, but a subadditive function is not necessarily anti-star-shaped.
% Let $x,y>0$. Then we have, 
% $$\frac{1}{x+y}\overline F\(\frac{1}{x+y}\)\le \frac{1}{x}\overline F\(\frac{1}{x}\) \iff \frac{x}{x+y}\overline F\(\frac{1}{x+y}\)\le\overline F\(\frac{1}{x}\),$$
% and 
% $$\frac{1}{x+y}\overline F\(\frac{1}{x+y}\)\le \frac{1}{y}\overline F\(\frac{1}{y}\)\iff \frac{y}{x+y}\overline F\(\frac{1}{x+y}\)\le\overline F\(\frac{1}{y}\).$$
% Summing up the inequalities on the right-hand side, we have   
% $$\overline F\(\frac{1}{x+y}\)\le \overline F\(\frac{1}{x}\)+\overline F\(\frac{1}{y}\).$$
Hence, $\mathcal{CS}\subsetneq\mathcal{IS}$.

For (ii), suppose $X\in \mathcal{M}$. Then $y\log F(y)$ decreases in $y>0$, which implies $-\log F(1/u)/u$ decreases in $u>0$. Thus,  $-\log F(1/u)$ is anti-star-shaped, which implies subadditivity and $\mathcal{M}\subseteq \mathcal{D}$. Furthermore, $\mathcal{D}\subsetneq \mathcal{IS}$ has been shown by Theorem 4.13 and Example 4.14 of \cite{ALO25}. For $\mathcal {CS}\not\subseteq\mathcal{D}$, it is easy to check the distribution function $F(x)=1-1/x$ for $x\ge 1$ belongs to $\mathcal{CS}$ but does not belong to $\mathcal{D}$. Example 3.6 of \cite{ZZSH25} provides a distribution that is in $\mathcal{D}$ but is not in $\mathcal{CS}$.

To see $\mathcal{M}\subsetneq \mathcal{D}$ and $\mathcal{M}\subsetneq \mathcal{CS}$, we provide a distribution that is in $\mathcal{D}$ and $\mathcal{CS}$ but not in $\mathcal{M}$. Let 
$$
F(x)=
\begin{cases}
\exp(-4),~~\mbox{for $0\le  x< 1/3$},\\
\exp(2-2/x), ~~\mbox{for $1/3\le x< 1/2$},\\
\exp(-2), ~~\mbox{for $1/2\le x<1$},\\
\exp(-2/x), ~~ \mbox{for $ x\ge 1$}.
\end{cases}$$
Let $g(t)=-\log F(1/t)$ for $t >0$. We have 
$$
g(t)=
\begin{cases}
2t,~~\mbox{for $0<  t< 1$},\\
2, ~~\mbox{for $1\le t< 2$},\\
2t-2, ~~\mbox{for $2\le t<3$},\\
4, ~~ \mbox{for $t\ge 3$}.
\end{cases}$$
To prove $F\in\mathcal D$, it remains to show that $g$ is subadditive.
Let $a,b>0$, and assume without loss of generality that $a\le b$. If
$a+b\ge 3$, the claim is immediate when $a\ge 3$. If $a<3$, then
$b\ge 3-a$; since $g$ is increasing and $g(t)+g(3-t)=4$ for $0<t<3$,
$$
g(a)+g(b)\ge g(a)+g(3-a)=4=g(a+b).
$$
Now suppose $a+b<3$. If $b\ge a\ge 1$, 
$g(a)+g(b)\ge 4>g(a+b)$. Assume $a<1$. If $b<1$, then
$g(a+b)\le 2(a+b)=g(a)+g(b)$. If $1\le b<2$, then either $a+b<2$, in
which case $g(a+b)=2\le 2a+2=g(a)+g(b)$, or $2\le a+b<3$, in which case
$$
g(a+b)=2(a+b)-2\le 2a+2=g(a)+g(b).
$$
Finally, if $2\le b<3$, then
$g(a+b)=2(a+b)-2=2a+(2b-2)=g(a)+g(b)$. Hence, $g$ is subadditive and $F \in \mathcal{D}$.
Let $h(t)=\bar F(1/t)/t$ for $t>0$. We have $h(t)=(1-e^{-g(t)})/t$ and we can easily check that $h$ is continuous and decreasing on $(0,1)$, $(1,2)$, $(2,3)$ and $(3, \infty)$. Therefore, $x\bar F(x)$ is increasing and $X \in \mathcal{CS}$. However, $F \notin \mathcal{M}$ as $\log F(1/3)/3=-4/3<\log(F(1/2))/2=-1$. Therefore, $x\log F(x)$ is not decreasing.
% As $\mathcal {D}_3\subseteq \mathcal {CS}\subseteq \mathcal{IS}$ and $\mathcal {D}_3\subseteq  \mathcal{D}$, $\mathcal {CS}, \mathcal{IS}$, and $\mathcal{D}$ are not subsets of $\mathcal{S}_C$.

 For (iv), since the Cauchy distribution is in $\mathcal{S}_C$ but is not in $\mathcal{CS}$, $\mathcal{M}$, $\mathcal{IS}$, and $\mathcal{D}$,  $\mathcal{S}_C$ is not a subset of all the latter sets. The following distribution $G$ is in $\mathcal{M}$,  which is given by 
$$
G(x)=
\begin{cases}
0,~~\mbox{for $x\le 0$},\\
\exp\(-1/x\),~~\mbox{for $0< x<1$},\\
\exp(-1), ~~\mbox{for $1\le x<2$},\\
\exp(-2/x), ~~\mbox{for $x\ge 2$}.
\end{cases}$$
However, it is not in $\mathcal{S}_C$ since $\tan(\pi G(x)-\pi/2)$ is not concave. Therefore, $\mathcal {M}\not\subseteq \mathcal{S}_C$. 
 \end{proof}

\begin{figure}
\centering
\begin{tikzpicture}[scale=1.5, every node/.style={font=\small}]

% The largest class IS
%\fill[purple!12] (0,0) ellipse (3.35 and 2.35);
\draw[thick] (0,0) ellipse (3.35 and 2.35);
\node at (0,2.05) {$\mathcal{IS}$};

% The two intermediate classes
%\fill[blue!18] (-1,0) circle (1.65);
%\fill[red!18]  ( 1,0) circle (1.65);

\draw[thick] (-0.5,0) circle (1.9);
\draw[thick]  (0.5,0) circle (1.9);

% The smaller class M inside the intersection
%\fill[green!25] (0,0) circle (0.55);
\draw[thick] (0,0) circle (1.3);
\draw[thick] (0,0) circle (0.8);
\draw[thick] (0,0) circle (0.3);
% Labels
\node at (-2,0.5) {$\mathcal D$};
\node  at ( 2,0.5) {$\mathcal{CS}$};
\node at (0,0) {$S_P^+$};
\node at (0.5,0.3) {$S_F^+$};
\node at (0.9,0.65) {$\mathcal{M}$};
\end{tikzpicture}
\caption{Relationship of  $S_P^+$, $S_F^+$, $\mathcal{M}$, $\mathcal{D}$, $\mathcal{CS}$ and $\mathcal{IS}$}\label{fig:venn}
\end{figure}
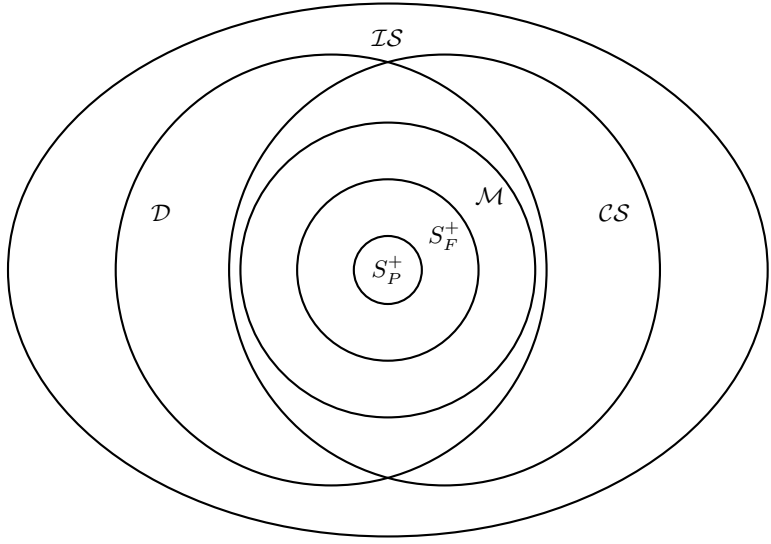

Figure \ref{fig:venn} shows the relationship between non-negative risk classes, including $S_P^+$, $S_F^+$, $\mathcal{M}$, $\mathcal{D}$, $\mathcal{CS}$ and $\mathcal{IS}$. Theorem \ref{th:convex} implies WUVS for random vectors with independent but non-identically distributed components in $\mathcal S_P$, $\mathcal  S_F$, and $\mathcal  S_C$ classes. Theorem \ref{pro:SAw} further provides strict WUVS for random vectors with independent but non-identically distributed components in strict $S_P$, $S_F$ and $S_C$ classes. Clearly, Theorem \ref{thm:cs} enlarges the class of risks $\mathcal M$ considered by \cite{M25} for UVS, in the case of independent marginals. 
 % Hence, \eqref{eq:SD} holds for independent risks $X_1$ and $X_2$ that are Fr\'echet$(1)$,  Pareto$(1)$ or Cauchy distribution. 
 To fully understand UVS for the risk classes in Table \ref{tab:example}, we need to check UVS for random vectors with
  independent  components from   $\mathcal{D}$ or $\mathcal{IS}$.
The following example shows that UVS does not hold for all independent risks from $\mathcal{D}$ or $\mathcal{IS}$ without assuming identical marginals. 

% \begin{example}\label{ex:counter-D}
% We consider  the following distribution from \citet[Example 2.7]{ALO25}
% $$F_p(x)=(1-p)^{\lceil 1/x\rceil}, ~~~ x>0, ~p\in (0,1).$$
% We can check that $-\log (F_p(1/x))=-\lceil x\rceil\log(1-p)$ is subadditive  as $\lceil x\rceil$ is subadditive and $\log(1-p)<0$. Hence, $F_p \in \mathcal{D}$.
% An $F_p$ distributed risk is a discrete random variable with $\p(X=1/k)=p(1-p)^k$, $k \in \N$,  and $\p(X=\infty)=p$. \com{I think in our theory we do not allow random variables to take the value $\infty$, so some changes are needed here.}
% Let $X \sim F_p$, $Y \sim F_q$ and $(X^c, Y^c)$ be a comonotonic copy of $(X, Y)$. 
% Take $p=0.1$ and $q=0.05$. Let $U$ be a uniform random variable on $(0,1)$. We can check that 
% $$\p\left(X^c+Y^c>\frac{5}{4}\right)
% =\p\left(U>\left(\frac{19}{20}\right)^4\right)=1-\left(\frac{19}{20}\right)^4\ge 0.18$$
% and
% \begin{align*}
% &\p\left(X+Y>\frac{5}{4}\right)\\ &=\p\left(\{X=\infty\}\cup \{Y=\infty\}\right)+\p\left(X=1, Y\in \left\{1,\frac{1}{2}, \frac{1}{3}\right\}\right)+\p\left(X \in \left\{\frac{1}{2}, \frac{1}{3}\right\}, Y=1\right)\\
% &\le 0.17\le 0.18\le \p\left(X^c+Y^c>\frac{5}{4}\right).
% \end{align*}
% Hence, $ X_1^c+ X_2^c \not\le_{\rm st}  X_1+X_2$ and $(X_1, X_2) \notin \mathcal{L}^{\rm U}_2$.
% \end{example}

\begin{example}\label{ex:counter-D}
Let $k=\lfloor t/L\rfloor$ and $$h_{L,c}(t)=\inf_{m \in \mathbb{N}_0}\{mc+(t-mL)_+\}=ck+(t-k L)\wedge c, ~~~t>0,$$
where   $0<c<L$  and $\N_0=\N\cup\{0\}$. For any $m,n \in \N_0$ and $t,s >0$, we have 
$$(m+n) c+(t+s-(m+n)L)_+\le mc+(t-mL)_++ nc+(s-nL)_+.$$
Hence, for $t,s>0$, 
$$h_{L,c}(t+s)=\inf_{\ell \in \mathbb{N}_0}\{\ell c+(t+s-\ell L)_+\}\le mc+(t-mL)_++ nc+(s-nL)_+.$$
As the above inequality holds for any $m,n \in \N_0$,  it is clear that $h_{L,c}$ is subadditive. Moreover,  $h_{L,c}$ is continuous, increasing, and it takes values in $(0, \infty)$.
Consider  $$F_{L,c}(x)=\exp\{-h_{L,c}(1/x)\},~~~x>0.$$  
By the continuity and monotonicity of $h_{L,c}$,  $F_{L,c}$ is increasing and continuous. Moreover,  $\lim_{x\to 0} F_{L,c}(x)=0$ and  $\lim_{x\to  \infty} F_{L,c}(x)=1$. Hence, $F_{L,c}$ is a cumulative distribution function of a continuous random variable in $(0, \infty)$. Since $-\log(F_{L,c}(1/x))=h_{L,c}(x)$ is subadditive, $F_{L,c} \in \mathcal{D}$. 
Next, consider independent risks $X_1\sim F_{1, 1/500}$ and $X_2\sim F_{10,1/5}$.
Direct calculations (omitted here) show that  $\VaR_p(X_1+X_2)< \VaR_p(X_1)+\VaR_p(X_2)$ holds at $p=0.68$. Hence, $ X_1^c+ X_2^c \not\le_{\rm st}  X_1+X_2$ and $(X_1, X_2) \notin \mathcal{L}^{\rm U}_2$. 
 Figure \ref{fig:ex6} plots $\VaR_p(X_1+X_2)$ and $\VaR_p(X_1)+\VaR_p(X_2)$ over $p \in (0, 0.8)$ and $p\in (0.8, 0.999)$ in two subfigures for better visibility. 
\begin{figure}[h]
\centering
\includegraphics[width=15cm]{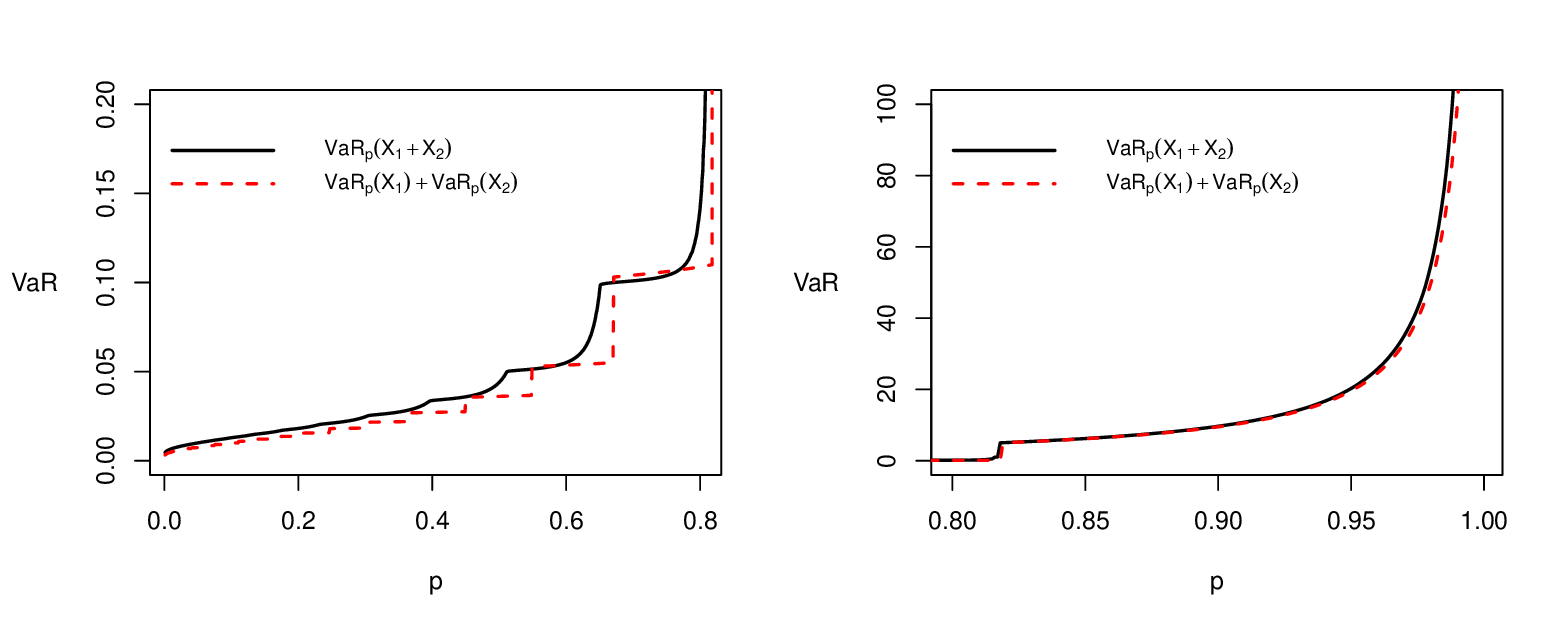}
\caption{$\VaR_p(X_1+X_2)$ and $\VaR_p(X_1)+\VaR_p(X_2)$ for $p\in(0,0.8)$ (left) and  $p\in(0.8,0.999)$ (right), where $X_1\sim F_{1, 1/500}$ and $X_2\sim F_{10,1/5}$ are independent.}    \label{fig:ex6}
\end{figure}
\end{example}

\section{Conclusion}\label{sec:con}

This paper provides the first systematic study of universal VaR superadditivity (UVS) as a property of random vectors, and the related notion of weighted universal VaR superadditivity (WUVS). We establish a necessary condition (Proposition \ref{prop:como}) as well as an equivalent stochastic dominance representation of UVS. Moving beyond the predominantly iid framework in the literature, Theorem \ref{th:convex} provides a useful approach for generalizing the existing results: If a random vector $(X_1,\dots,X_n)$ satisfies (strict) WUVS,  $(f_1(X_1),\dots,f_n(X_n))$ also satisfies (strict) WUVS, provided that $f_1,\dots,f_n$ are (strictly) increasing and convex functions. A similar statement holds for UVS with identical transforms.  Theorem \ref{pro:SAw} further provides another way to extend nonstrict UVS and WUVS to strict  UVS and WUVS by strictly convex transformations, and Theorem \ref{th:closure} obtains closure properties. 
Theorem \ref{thm:cs} shows UVS and WUVS for independent risks from the completely subscalable class of \cite{V25}.
Theorem \ref{pro:SAw} helps build UVS and WUVS for random vectors with risks that fall into three broad classes: the super-Cauchy class of \cite{Muller25}, a subclass of the inverted subadditive class of \cite{ALO25}, and a subclass of the $\mathcal D$ class of \cite{CS25}. {Moreover, we further establish strict WUVS for many distributions, which leads to stronger implications in decision making, whereas the literature mainly focuses on  non-strict inequalities.}
Our results are mainly theoretical in nature, aiming to deeply understand portfolio vectors that satisfy UVS and WUVS, but  they have useful implications  in risk management and decision making, as illustrated by Proposition \ref{prop:rm}.

% Except for \cite{M25}, most existing studies of UVS and WUVS consider random vectors with identically distributed components. 

 %two classes of risks in Theorem \ref{thm:cs} and Corollary \ref{thm:IS} are incomparable. First, since $\mathcal{CS}\subset \mathcal{IS}$, an $\mathcal{IS}$ risk may not be a $\mathcal{CS}$ risk. Second, by Theorem \ref{thm:cs},  \eqref{eq:SD} holds for independent risks $X_1$ and $X_2$ that are Fr\'echet$(1)$ and Pareto$(1)$.  %However, this cannot be concluded from  Theorem \ref{thm:IS} since for an arbitrary risk $X_i$ in Theorem \ref{thm:IS}, it has the representation $X_i\laweq f_i(Y)$ where $Y\in\mathcal{IS}$ and $f_i$ is a strictly increasing and convex function with $f_i(0)=0$. Such representations are impossible for Fr\'echet$(1)$ and Pareto$(1)$ risks which have distinct essential infimums 0 and 1, since if the essential infimum of $Y$ is equal to (resp.~strictly larger than) 0, the essential infimum of $f_i(Y)$ is also equal to (resp.~strictly larger than) 0.

\end{document}